\documentclass{mytemplate}

\usepackage{ulem}

\title{Efficient and Debiased Learning of Average Hazard Under Non-Proportional Hazards}

\author{
Xiang Meng\thanks{Corresponding author: \texttt{xmeng@g.harvard.edu}. Postdoctoral Fellow at Dana-Farber Cancer Institute.}
\and
Lu Tian \thanks{Professor, Department of Biomedical Data Science, Stanford University}
\and
Kenneth Kehl\thanks{Assistant Professor of Medicine, Harvard Medical School and Dana-Farber Cancer Institute}
\and
Hajime Uno\thanks{Associate Professor,  Harvard Medical School and Department of Data Science, Dana-Farber Cancer Institute}
}

\begin{document}
\maketitle

\begin{abstract}
The hazard ratio from the Cox proportional hazards model is a ubiquitous summary of treatment effect. However, when hazards are non-proportional, the hazard ratio can lose a stable causal interpretation and become study-dependent because it effectively averages time-varying effects with weights determined by follow-up and censoring. We consider the average hazard (AH) as an alternative causal estimand: a population-level person-time event rate that remains well-defined and interpretable without requiring proportional hazards. Although AH can be estimated nonparametrically and regression-style adjustments have been proposed, existing approaches do not provide a general framework for flexible, high-dimensional nuisance estimation with valid $\sqrt{n}$ inference. We address this gap by developing a semiparametric, doubly robust framework for covariate-adjusted AH. We establish pathwise differentiability of AH in the nonparametric model, derive its efficient influence function, and construct cross-fitted, debiased estimators that leverage machine learning for nuisance estimation while retaining asymptotically normal, $\sqrt{n}$-consistent inference under mild product-rate conditions. Simulations demonstrate that the proposed estimator achieves small bias and near-nominal confidence-interval coverage across proportional and non-proportional hazards settings, including crossing-hazards regimes where Cox-based summaries can be unstable. We illustrate practical utility in comparative effectiveness research by comparing immunotherapy regimens for advanced melanoma using SEER-Medicare linked data.
\end{abstract}

\section{Introduction}

Time-to-event endpoints are the cornerstone of clinical trials and observational comparative effectiveness research. The Cox proportional hazards model and its hazard ratio (HR) remain ubiquitous summaries of treatment effect \citep{cox1972regression}. However, the proportional hazards assumption is frequently violated in modern applications, with departures such as delayed treatment effects that are common in immuno-oncology \citep{Alexander2018-iq} or crossing hazards routinely diagnosed in practice \citep{grambsch1994proportional}. Even in randomized trials, the Cox HR compares instantaneous hazards within time-specific risk sets among individuals who have survived event-free up to each time point. Because risk sets are post-treatment--selected, the HR is a non-collapsible contrast and can be difficult to interpret causally and clinically when hazards vary over time \citep{hernan2010hazards,aalen2015cox}.

When hazards are non-proportional, these interpretational difficulties are amplified. The Cox partial likelihood estimator no longer targets a single, time-invariant parameter; instead, it converges to a weighted average of the time-varying log hazard ratio. Crucially, the weights depend on the study-specific censoring and follow-up distributions \citep{struthers1986misspecified,xu2000estimating,schemper2009averageHR,Horiguchi2019-xc}. Consequently, two studies with identical biological treatment effects can yield materially different HR estimates simply because their accrual rates or censoring patterns differ. Moreover, as a relative measure lacking an absolute reference value, the HR can be difficult to translate into clinically actionable information for patients and decision-makers \cite{uno2014moving}.

These concerns have motivated greater emphasis on estimands defined directly through the marginal survival curve---for example, survival differences at fixed horizons and contrasts of restricted mean survival time (RMST) \citep{royston2013rmst,uno2014moving}. These quantities provide absolute reference values and remain well-defined without assuming proportional hazards. They are often directly interpretable for patients and clinicians, while complementing rate-based summaries on the hazard scale.

To provide an interpretable hazard-scale between-group summary that does not rely on proportional hazards, the average hazard (hereafter AH) \citep{uno2023ratio}, also referred to as the generalized hazard \citep{Snapinn2022-wf,Jackson2025-vq}, has been recently proposed. 

For a treatment arm $a$ and horizon $\tau$, the AH is defined as
\begin{equation} \eta_a(\tau) 
= \frac{\Pr(T_a \le \tau)}{\mathbb{E}[T_a \wedge \tau]} 
= \frac{1-S_a(\tau)}{\int_0^\tau S_a(t)\,dt}
= \frac{\int_0^\tau h_a(t)S_a(t)\,dt}{\int_0^\tau S_a(t)\,dt}.
\label{AH1}
\end{equation} 
where $h_a(t)$ and $S_a(t)$ are hazard function and survival function, respectively. 
This representation shows the AH as a population-level person-time event rate: time periods contribute to the average in proportion to the probability of the population remaining at risk. For example, an AH of 0.05 indicates that, on average, one would expect 5 events per 100 person-years of follow-up within the window $[0, \tau]$.

AH is conceptually related to the conventional person-time incidence rate (IR), often described as {\it `the number of observed events divided by the total observation time'} \citep{rothman2008modern}. 
Under general censoring, the IR as defined by \cite{rothman2008modern} can be written as
\begin{equation} 
IR 
= \frac
  {\sum_{i=1}^{n}I(T_i \le C_i)}
  {\sum_{i=1}^{n}( T_i \wedge C_i )},
\label{IR}
\end{equation}
where $n$ denotes the sample size, and ${T_i}$ and ${C_i}$, $i=1,\ldots,n$, are assumed to be independent copies of the event time $T$ and censoring time $C$, respectively. 
As $n \to \infty$, (\ref{IR}) converges to
\begin{equation} 
 \frac
  {\Pr(T \le C) }
  {\mathbb{E}(T \wedge C ) }.
\label{IRlimit}
\end{equation}
The limit in (\ref{IRlimit}) is generally a function of the censoring-time distribution, except in special cases, for example when $T$ is exponentially distributed. The AH in (\ref{AH1}), however, is defined solely in terms of the event-time distribution. It can therefore be viewed as a censoring-independent person-time rate for $T$, and treatment contrasts such as the log AH ratio yield rate-based summaries that remain interpretable without requiring strong model assumptions such as proportional hazards \citep{uno2023ratio,horiguchi2026}.


The original development of AH focused on unadjusted two-sample comparisons \citep{uno2023ratio}. In many modern applications, however, baseline covariate adjustment is important. It can improve precision in randomized trials and is needed for identifying causal effects in observational studies where treatment assignment is confounded. Recent work has proposed stratified and regression-based approaches for AH \citep{Qian2025-bd,uno2024regressionAH} and has demonstrated AH-based analyses in observational comparative effectiveness research \citep{Xiong2025-bi}. 
Nevertheless, existing AH methods do not provide a general covariate-adjusted causal inference framework that simultaneously (i) targets causal AH contrasts in observational settings by explicitly accounting for confounding and censoring, (ii) accommodates flexible, data-adaptive estimation of treatment, censoring, and outcome nuisance functions, and (iii) yields  valid $\sqrt{n}$ inference when these nuisance functions are estimated using machine learning methods.

We address this gap by developing a semiparametric, doubly robust framework for covariate-adjusted AH. Our contributions are three-fold.
First, we formulate the AH as a causal estimand using potential outcomes and establish identification under standard assumptions. 
Second, we derive the efficient influence function for $\log \eta_a(\tau)$ by combining the efficient influence function for the survival curve with a functional delta method. 
Third, we construct cross-fitted debiased estimator with doubly robust remainder structure, allowing flexible machine learning estimation of nuisance functions while preserving asymptotically normal, $\sqrt{n}$-consistent inference under mild rate conditions \citep{chernozhukov2018double,westling2024inference,cai2020one}. 
We evaluate the finite-sample performance of the proposed estimator through simulation and demonstrate its practical utility by analyzing clinical outcomes among older adults with advanced melanoma using SEER-Medicare data, a large-scale linkage of population-based cancer registries and longitudinal administrative claims commonly used for real-world comparative effectiveness research. The application examines whether first-line nivolumab plus ipilimumab or nivolumab monotherapy improves overall survival compared with ipilimumab monotherapy.

Section~\ref{sec:causal_quantity} defines the causal AH estimand and its identification. Sectio~\ref{sec:debiased-estimation} derives the efficient influence function and details the construction of the debiased estimator. Section~\ref{sec:large-sample-theory} provides the asymptotic theory. Sections~\ref{sec:simulation} and \ref{sec:data-analysis} present simulation studies and the data analysis, respectively.

\section{Causal Quantity: Average Hazard}
\label{sec:causal_quantity}

In this section, we define AH in terms of potential outcomes, describe the identifying assumptions, and provide an identification theorem linking the causal parameter to the observed data distribution. 

\subsection{Causal definition}

Let $T$ denote the event time of interest, $A \in \{0,1\}$ a binary treatment indicator, and $W \in \mathcal{W} \subset \mathbb{R}^d$ a vector of baseline covariates. To define treatment effects, we adopt the potential outcomes framework \citep{rubin1974estimating, pearl2009causality}, where $T(a)$ denotes the counterfactual (or potential) event time that would have been observed for an individual had they received treatment $a$. 

For a fixed truncation time $\tau > 0$, we define the counterfactual survival function as the probability that an individual's potential event time under treatment $a$ exceeds $t$:
\[
S_0(t \mid a) = \Pr\{T(a) > t\}, \qquad t \in [0,\tau].
\]
This function represents the survival experience of the entire population if everyone were assigned to treatment $a$, effectively removing the influence of confounding and selection bias. From this, we define two key summary measures:
\begin{equation}
\label{eq:F-tau-R-tau}
F_a(\tau) = \Pr\{T(a)\le \tau\}=1-S_0(\tau \mid a),
\qquad
R_a(\tau)
= \mathbb{E}\!\left[T(a)\wedge\tau\right]
= \int_0^\tau S_0(t \mid a)\,dt,
\end{equation}
where $F_a(\tau)$ is the counterfactual cumulative incidence (the probability of the event occurring by time $\tau$) and $R_a(\tau)$ is the counterfactual restricted mean survival time (RMST), representing the expected time lived free of the event up to horizon $\tau$. The identity $\mathbb{E}[T(a)\wedge\tau]=\int_0^\tau \Pr\{T(a)>t\}\,dt$ follows from the tail-integral representation
$T(a)\wedge\tau=\int_0^\tau \mathbb{I}\{T(a)>t\}\,dt$ and Tonelli's theorem (see, e.g., \cite{royston2013rmst, uno2014moving}).

\begin{definition}[Average hazard]
The average hazard under treatment $a$ and truncation time $\tau$ is defined as the ratio:
\begin{equation}
\label{eq:AH}
    \eta_a(\tau) = \frac{F_a(\tau)}{R_a(\tau)}.
\end{equation}
\end{definition}

Equivalently, using the identity 
\(F_a(\tau)=\int_0^\tau S_a(t)h_a(t)\,dt\), the AH can be written as
\[
\eta_a(\tau)
=
\frac{\int_0^\tau h_a(t)S_a(t)\,dt}
{\int_0^\tau S_a(t)\,dt}
=
\int_0^\tau h_a(t)
\frac{S_a(t)}{\int_0^\tau S_a(u)\,du}\,dt .
\]
Thus, AH is a survival-weighted average of the hazard, with weights proportional to the amount of event-free person-time contributed by the population at each time.

By contrast, the unweighted time average
\(\tau^{-1}\int_0^\tau h_a(t)\,dt\) summarizes the cumulative hazard per unit calendar time. This is also a well-defined hazard-scale summary, but it averages over calendar time rather than over the event-free person-time contributed by the population. Hence, it has a different interpretation from AH and should not be viewed as the same population-level person-time rate.

The AH is an \emph{absolute}, arm-specific summary measure with a direct person-time interpretation. It is defined as a functional of the marginal survival distribution within a treatment arm and can therefore be estimated separately for each treatment regimen. Between-group contrasts, such as the ratio of average hazards (RAH) or difference in average hazards (DAH), are then formed by comparing these arm-specific summaries.
This construction differs from the Cox hazard ratio, which is itself a relative contrast defined through instantaneous hazards within the time-$t$ risk set, that is, among individuals who have survived event-free up to $t$. Because these survivor sets are post-treatment and may differ systematically between arms, the Cox hazard-ratio contrast can be difficult to interpret causally even under randomization \citep{hernan2010hazards,aalen2015cox}. In contrast, arm-specific AHs are functionals of the marginal counterfactual survival distributions, and AH-based contrasts therefore admit a population-level causal interpretation under the standard identification assumptions stated below.


We take $\eta_a(\tau)$ as our causal estimand. Throughout this paper, we focus on the log average hazard ratio between arms:
\begin{equation}
\label{eq:log-RAH}
    \beta(\tau) = \log\!\left(\frac{\eta_1(\tau)}{\eta_0(\tau)}\right).
\end{equation}
While other contrasts are possible, we focus on the log-ratio scale as our primary analytical target. Operationally, we derive the inference procedure for $\beta(\tau)$, though results can be readily exponentiated to report the Ratio of Average Hazards (RAH) and its corresponding confidence intervals which provides a familiar scale for clinical interpretation \citep{uno2023ratio}.

Several properties distinguish AH from traditional Cox's hazard ratio (HR) and motivate its use as a causal estimand.

\begin{remark}[Study-design dependence of the Cox target under non-proportional hazards]

In an AH-based approach, treatment effects are summarized by comparing arm-specific AHs, for example through the ratio or difference of AHs. These AH-based contrasts are defined directly from the marginal counterfactual survival distributions and therefore do not rely on strong modeling assumptions such as proportional hazards. 
By contrast, the conventional Cox approach summarizes treatment effects through a proportional hazards working model. When that model is misspecified, the resulting coefficient generally targets a weighted average of time-varying hazard ratios rather than a single model-free treatment-effect estimand \citep{struthers1986misspecified,schemper2009averageHR,Tian2017}. Specifically, 
\[
\beta^\ast
\approx
\frac{
\int_0^\tau w(u)\lambda_0(u)
\log\{\lambda_1(u)/\lambda_0(u)\}\,du
}{
\int_0^\tau w(u)\lambda_0(u)\,du
},
\qquad
w(u)=
\frac{p_0p_1\pi_0(u)\pi_1(u)}
{p_0\pi_0(u)+p_1\pi_1(u)},
\]
where \(\pi_j(u)=\Pr(T_j\wedge C_j\ge u)\). 
Because these weights depend on the risk-set and censoring process, the limiting Cox estimand can vary across studies with identical underlying event-time distributions but different accrual, administrative censoring, or loss-to-follow-up patterns. 
In contrast, AH-based contrasts are constructed from 
\[
\eta_a(\tau)=\frac{1-S_a(\tau)}{\int_0^\tau S_a(t)\,dt},
\]
and therefore compare population-level person-time event rates that are not altered by such study-specific design features, provided censoring is independent and identification assumptions hold.
\end{remark}

\begin{remark}[Consistency with stochastic ordering]
Finally, AH respects stochastic ordering. If $S_1(t) \ge S_0(t)$ for all $t \le \tau$ (i.e., survival under treatment 1 stochastically dominates survival under treatment 0), then $\eta_1(\tau) \le \eta_0(\tau)$. This property ensures that AH does not yield paradoxical comparisons when survival curves are ordered.
\end{remark}

\subsection{Identification of the Causal Parameter}
\label{subsec:identification}

We now link the causal parameter to the observed data distribution. The assumptions are similar to those made in \cite{westling2024inference}. It is useful to separate
(i) causal assumptions that identify the estimand as a functional of the full-data law of $(W,A,T)$
from (ii) additional assumptions that link this functional to the actually observed right-censored data
$(W,A,U,\Delta)$, where $U=\min(T,C)$, $\Delta=\mathbb{I}\{T\le C\}$ and $C$ is the censoring time. 

\begin{enumerate}
  \item[(A1)] \textbf{Consistency:} $T = T(a)$ if $A=a$.
  \item[(A2)] \textbf{Conditional exchangeability:} $(T(0),T(1)) \perp A \mid W$.
  \item[(A3)] \textbf{Positivity:} $0 < \pi_a(W) = \Pr(A=a\mid W) < 1$ almost surely.
  \item[(A4)] \textbf{Conditional independent censoring:} $C \perp T(a) \mid (A,W)$.
  \item[(A5)] \textbf{Censoring support:} $\Pr(C \ge \tau \mid A=a,W) > 0$ almost surely.
\end{enumerate}

For each arm $a$, define the conditional survival in the full-data world by
\[
S_0(t \mid a,w) \;:=\; \Pr(T>t \mid A=a, W=w),
\]
and the marginal covariate-adjusted survival curve by
\[
S_a(t) \;:=\; \mathbb{E}\!\left[ S_0(t \mid a, W) \right],\qquad t\in[0,\tau],
\]
where the expecation is with respect to the marginal distribution of $W.$
Under (A1), (A2), and (A3), the marginal curve $S_a(\cdot)$ coincides with the counterfactual survival
$t\mapsto \Pr\{T(a)>t\}$, and therefore identifies $\eta_a(\tau)$ via \eqref{eq:F-tau-R-tau}--\eqref{eq:AH}.

\begin{theorem}[Causal identification of AH]
\label{thm:identification}
Fix $a\in\{0,1\}$ and $\tau>0$. Under (A1)--(A3), the causal average hazard
$\eta_a(\tau)=F_a(\tau)/R_a(\tau)$ defined in \eqref{eq:AH} is identified by
\begin{equation}
\label{eq:single-arm-ah}
\eta_a(\tau)
=
\frac{1 - S_a(\tau)}{\int_0^\tau S_a(t)\,dt},
\qquad
S_a(t)=\mathbb{E}\big[S_0(t\mid a,W)\big],
\end{equation}
where $S_0(t\mid a,w)=\Pr(T>t\mid A=a,W=w)$.
\end{theorem}

\begin{proposition}[Observed-data identification under right censoring]
\label{prop:identification-censoring}
Suppose we observe $O=(W,A,U,\Delta)$ with $U=\min(T,C)$ and $\Delta=\mathbb{I}\{T\le C\}$.
Under (A4)--(A5), for each $(a,w)$ the conditional survival $S_0(t\mid a,w)$ is identified from the observed-data law.
In particular, the conditional cumulative hazard admits the observable representation
\[
\Lambda_0(t\mid a,w)
=
\int_0^t
\frac{\Pr(U\in [u,u+du],\Delta=1\mid A=a,W=w)}{\Pr(U\ge u\mid A=a,W=w)},
\]
so $S_0(t\mid a,w)=\pi_{(0,t]}\{1-d\Lambda_0(u\mid a,w)\}$ (or $S_0(t\mid a,w)=\exp\{-\Lambda_0(t\mid a,w)\}$ under continuity).
\end{proposition}

Proofs are given in Appendix~\ref{app:proof-identification}.




\begin{remark}[Comparison to an alternative marginalization approach] 
AH coincides with the true population-level person-time incidence rate. A seemingly natural alternative would be to average conditional incidence rates directly. We call this the marginalized conditional incidence rate, defined as
\begin{equation}
\label{eq:marg-cond-rate}
\tilde{\eta}_a(\tau) = \mathbb{E}\!\left[\frac{F(\tau \mid a,W)}{R(\tau \mid a,W)}\right],
\end{equation}
where $F(\tau \mid a,w)$ and $R(\tau \mid a,w)$ are conditional cumulative incidence and RMST, respectively.
Thus, $\tilde{\eta}_a(\tau)$ is the simple average of these conditional rates over $W$ \eqref{eq:marg-cond-rate}.

By contrast, AH (Equation~\ref{eq:AH}) can be expressed as a weighted average of conditional rates,
\[
\eta_a(\tau) = \mathbb{E}\!\left[ \frac{R(\tau \mid a,W)}{\mathbb{E}[R(\tau \mid a,W)]} \cdot \frac{F(\tau \mid a,W)}{R(\tau \mid a,W)} \right],
\]
with weight proportional to the contribution to the total person-time. Because the expectation of a ratio is not equal to the ratio of expectations, $\tilde{\eta}_a(\tau) \neq \eta_a(\tau)$ in general---a manifestation of non-collapsibility.

Interpretively, AH answers the question: “How many events per unit person-time occur in arm $a$ at the population level?” The marginalized conditional incidence rate answers: “What is the average of each subject’s expected event rate?” These are not the same. Only AH corresponds to the actual population-level incidence rate, whereas the marginalized conditional incidence rate can bias treatment contrasts by over-representing high-risk subgroups that contribute little person-time.

\end{remark}

In the Section~\ref{sec:debiased-estimation}, we build upon Theorem~\ref{thm:identification} to develop doubly robust and debiased estimators of $\eta_a(\tau)$, adapting the semiparametric framework of \citet{westling2024inference}.

\subsection{Related causal estimands for time-to-event outcomes.}
A growing literature has proposed alternative causal estimands for survival outcomes that avoid the interpretational difficulties of hazard-based contrasts. One prominent direction is the win ratio, net benefit, and related pairwise comparison estimands, which compare each treated subject to each control subject and tally wins, losses, and ties according to prespecified clinical priorities. These estimands can be attractive when multiple outcomes or event priorities are relevant, but they are intrinsically pairwise rather than a constrast between arm-specific functionals of a marginal survival function and therefore depend on design choices such as prioritization rules and tie handling. Right censoring further complicates interpretation: na\"ive win-ratio targets can become entangled with censoring unless the estimand is explicitly defined on the survival distributions only. Recent work has emphasized this estimand–estimator distinction and proposed censoring-robust causal definitions and inference for win-ratio-type targets \citep{mao2024defining, martinussen2025debiased}.

A second direction focuses on assumption-lean estimation of familiar regression targets in survival analysis. For example, \citet{dukes2024assumptionleancox} develop debiased inference for Cox regression parameters that remain well-defined under model misspecification and are compatible with flexible nuisance learning. While useful when a Cox-type coefficient is desired as a descriptive summary, such targets remain projections tied to risk-set comparisons and inherit non-collapsibility; under non-proportional hazards, their implicit weighting can still depend on censoring distributions. In contrast, the AH is defined directly as a smooth functional of the marginal counterfactual survival function, yielding a population-level person-time event rate that remains interpretable under non-proportional hazards and is invariant to any independent censoring.

Finally, several causal survival estimands target absolute risks or survival probabilities at fixed horizons. In particular, \citet{rytgaard2021onestep} consider intervention-specific risks evaluated at time $\tau$, which are g-computation functionals with clear causal interpretation under standard assumptions. These probability-scale estimands are natural when the specific time-point is of the primary interest. Our focus is complementary but distinct: we study a rate-based marginal estimand that combines absolute risk and RMST into an interpretable person-time event rate, and we develop efficient, doubly robust inference specifically for this AH and its contrasts.

\section{Debiased estimation}
\label{sec:debiased-estimation}
In this section, we derive efficient influence functions (EIFs) for the key building blocks of the AH estimand and construct a debiased, doubly robust estimator. 

Recall from Section~\ref{subsec:identification} (Theorem~\ref{thm:identification}) that under (A1)--(A5),
$S_a(t) = \mathbb{E}\!\left[S_0(t\mid a,W)\right]$ coincides with the counterfactual survival function $t\mapsto \Pr\{T(a)>t\}$. Moreover, by
\eqref{eq:F-tau-R-tau}--\eqref{eq:AH},
$F_a(\tau)=1-S_a(\tau)$, $R_a(\tau)=\int_0^\tau S_a(t)\,dt$, and $\eta_a(\tau)=F_a(\tau)/R_a(\tau)$.
Thus the arm-specific log average hazard $\theta_a(\tau)=\log\{\eta_a(\tau)\}$ is a smooth functional of the path
$t\mapsto S_a(t)$. We therefore first construct the EIF for $S_a(t)$ and then map it
to derive the EIF for $\theta_a(\tau)$ via the functional delta method. 

\subsection{Preliminaries: data structure and EIF for marginal survival}
\label{sec:eif-ah}

We observe i.i.d.\ data units $O=(W,A,U,\Delta)$, where $W\in\mathcal{W}$ are baseline covariates,
$A\in\{0,1\}$ is the treatment indicator, $U = \min\{T,C\}$, and $\Delta = \mathbb{I}\{T \le C\}$, with $C$ the censoring time.
Let $P_0$ denote the underlying distribution of $O$. 

For each treatment arm $a\in\{0,1\}$ we define the nuisance functions 
$\pi_0(a\mid w)=P_0(A=a\mid W =w),$ $
S_0(t\mid a,w)=P_0(T>t\mid A=a, W=w),$ $
G_0(t\mid a,w)=P_0(C\ge t\mid A=a,W=w),$  
and write $\Lambda_0(\cdot\mid a,w)$ for the conditional cumulative hazard function of $T$ given $(a,w)$, with associated
Riemann--Stieltjes measure $d\Lambda_0(\cdot\mid a,w)$. 

Before deriving the EIF for the AH, we first characterize the EIF for the treatment-specific marginal survival curve $S_a(t) = \mathbb{E}[S_0(t \mid a, W)]$, which serves as the primary building block for our estimand.

\begin{proposition}[Efficient influence function of $S_a(t)$, {\citet{westling2024inference}}]
\label{prop:eif-mu}
Under assumptions (A1)--(A5), the marginal survival curve $S_a(t)=\mathbb{E}[S_0(t\mid a,W)]$ is pathwise differentiable
in the nonparametric model with efficient influence function
\begin{align*}
\phi^*_{S_a(t)}(O)
&= \phi_{t,a}(O) - S_a(t), \\
\phi_{t,a}(O)
&= S_0(t\mid a,W)\Bigg\{
1 - \frac{\mathbb{I}\{A=a\}}{\pi_0(a\mid W)}
\Bigg[
\frac{\mathbb{I}\{U\le t,\ \Delta=1\}}{S_0(U\mid a,W)\,G_0(U\mid a,W)}
- \int_0^{t\wedge U} \frac{d\Lambda_0(u\mid a,W)}{S_0(u\mid a,W)\,G_0(u\mid a,W)}
\Bigg]
\Bigg\}.
\end{align*}
\end{proposition}

It is often convenient to rewrite the bracketed term in counting-process form. Define
\[
N(u) = \mathbb{I}\{U\le u,\ \Delta=1\},\qquad
Y(u) = \mathbb{I}\{U\ge u\},
\]
so that $M(t \mid a, w)=N(t)-\int_0^t Y(u)d\Lambda_0(u\mid a, W)$ is a martingale and $dM(u\mid a, w)=dN(u) - Y(u)\,d\Lambda_0(u\mid a,w)$ is a martingale increment. In this notation the bracketed term equals
\[
\int_0^{t\wedge U} \frac{dM(u\mid a, W)}{S_0(u\mid a,W)\,G_0(u\mid a,W)}.
\]
The factor $\mathbb{I}\{A=a\}/\pi_0(a\mid W)$ implements inverse probability of treatment weighting (IPTW), while
$1/G_0(u\mid a,W)$ implements inverse probability of censoring weighting (IPCW). The EIF therefore combines a centered
outcome-regression component with an IPTW--IPCW augmentation based on martingale residuals; see also
\citet{westling2024inference, zhang_doublerobust_2012}.

\subsection{Efficient influence function for log average hazard via functional delta method}
\label{subsec:eif-log-ah}

In the previous subsection, we characterized the properties of the marginal survival curve $S_a(t)$. Because $S_a(t)=\mathbb{E}\left[S_0(t \mid a, W)\right]$ is defined via an expectation with respect to the data generating process, we can view the entire curve as a functional of the data distribution $P$.

Proposition~\ref{prop:eif-mu} shows that the map $P \mapsto S_a(\cdot)$ is pathwise differentiable in the nonparametric model \citep[see][Chs.~20, 25]{vaart1998asymptotic}; see also Appendix~\ref{appendix:eif-ah}. With the EIF process identified as $t\mapsto\phi^*_{S_a(t)}(O)$, we now introduce the log-average-hazard map
\[
\Phi(S)
\;=\;
\log\!\left(\frac{1-S(\tau)}{\int_0^\tau S(u)\,du}\right),
\qquad S\in L^\infty([0,\tau]),
\]
and verify that $\Phi$ is Hadamard differentiable at $\mu_a$ in the sense of
\citet[][Ch.~20]{vaart1998asymptotic}.  With these two ingredients, a functional
delta method for pathwise differentiable parameters
\citep[][]{vaart1998asymptotic,bickel1993efficient}
yields the EIF of the composition
$\theta_a = \Phi\{S_a(\cdot)\}$ as the image of the survival EIF under the derivative
of $\Phi$.

Carrying out this program leads to the following result.

\begin{theorem}[Efficient influence function of log AH via survival]
\label{thm:eif-ah}
Under assumptions (A1)--(A5), the arm-specific log average hazard
$\theta_a(\tau)=\log\!\big(F_a(\tau)/R_a(\tau)\big)$
(Equations~\eqref{eq:F-tau-R-tau}--\eqref{eq:AH}) is pathwise
differentiable in the nonparametric model with efficient influence function
\begin{equation}
\label{eq:eif-ah-from-mu}
\phi^*_{\theta_a}(O)
= -\,\frac{1}{F_a(\tau)}\,\phi^*_{S_a(\tau)}(O)
\;-\;\frac{1}{R_a(\tau)}\int_0^\tau \phi^*_{S_a(u)}(O)\,du,
\end{equation}
where $\phi^*_{S_a(t)}(O)$ is given in Proposition~\ref{prop:eif-mu}.
\end{theorem}

A detailed proof, including formal definitions of pathwise and Hadamard
differentiability and the precise delta-method chain rule we use, is
provided in Appendix~\ref{appendix:eif-ah}.

\begin{remark}[Structure of the log-AH EIF]
\label{rem:structure-ah-EIF}
Substituting Proposition~\ref{prop:eif-mu} into~\eqref{eq:eif-ah-from-mu} yields the fully expanded form:
\begin{align*}
\phi^*_{\theta_a}(O)
&= -\frac{1}{F_a(\tau)}\!\left\{S_0(\tau\mid a,W)-S_a(\tau)\right\}
   -\frac{1}{R_a(\tau)}\!\int_0^\tau\!\left\{S_0(u\mid a,W)-S_a(u)\right\}\,du \\
&\quad +\frac{\mathbb{I}\{A=a\}}{\pi_0(a\mid W)}
\left[
\frac{S_0(\tau\mid a,W)}{F_a(\tau)}\int_0^\tau \frac{dM(u \mid a,W)}{S_0(u\mid a,W)\,G_0(u\mid a,W)} \right.\\
&~~~~~~~~~~~~~~~~~~~~~~\left.+\frac{1}{R_a(\tau)}\int_0^\tau\Bigg(\int_u^\tau S_0(t\mid a,W)\,dt\Bigg)\,
\frac{dM(u\mid a,W)}{S_0(u\mid a,W)\,G_0(u\mid a,W)}
\right].
\end{align*}
Writing $R_0(t\mid a,W):=\int_0^t S_0(v\mid a,W)\,dv$, the second term  can be expressed as
\[
+\frac{\mathbb{I}\{A=a\}}{\pi_0(a\mid W)}
\int_0^\tau 
\underbrace{\Bigg[\frac{S_0(\tau\mid a,W)}{F_a(\tau)}
+\frac{R_0(\tau\mid a,W)-R_0(u\mid a,W)}{R_a(\tau)}\Bigg]}_{\text{\normalsize time-dependent weight:} w_a(u, W)}
\frac{dN(u)-Y(u)\,d\Lambda_0(u\mid a,W)}
{S_0(u\mid a,W)\,G_0(u\mid a,W)}.
\] 
The first two terms form a centered outcome-regression component for $F_a(\tau)$ and $R_a(\tau)$, while the
third term is an IPCW-weighted martingale term. The time-dependent weight $w_a(u,W)$ emphasizes event times
$u$ that contribute more to both the cumulative incidence $F_a(\tau)$ and the RMST $R_a(\tau)$, making the connection
between log-AH and person-time contributions explicit.
\end{remark}

\begin{remark}[Inheritance of double robustness]
\label{rem:dr-ah-EIF}
Because $\phi^*_{S_a(t)}$ is doubly robust in the sense of \citet{westling2024inference}, the mapped EIF $\phi^*_{\theta_a}$ inherits the same properties. Specifically, the estimating equation based on $\phi^*_{S_a(t)}$ has mean zero if either the outcome model $S_0$ is correctly specified, or both the treatment mechanism $\pi_0$ and censoring survival $G_0$ are correctly specified (assuming mild rate conditions under cross-fitting). Consequently, the estimator for $\theta_a(\tau)$ remains consistent under the same set of identification conditions.
\end{remark}

\subsection{Cross-fitted plug-in estimator for log AH via efficient survival}
\label{sec:cf-ah}

We estimate the log average hazard (log-AH) using a plug-in approach based on an efficient estimator of the arm-specific marginal survival curve. Specifically, we construct the cross-fitted one-step estimator for the survival path $t\mapsto S_a(t)=\mathbb{E}[S_0(t\mid a,W)]$ and then apply the smooth map defined in Section~\ref{sec:causal_quantity}. Under standard regularity conditions, this survival estimator is asymptotically linear with influence function equal to the canonical gradient given in Proposition~\ref{prop:eif-mu}. By the functional delta method, the resulting log-AH estimator is asymptotically efficient, and its variance is consistently estimated by the empirical variance of the mapped EIF in Theorem~\ref{thm:eif-ah}.

\medskip
\noindent\textbf{Inputs.} Data $\{O_i=(W_i,A_i,U_i,\Delta_i): i=1,\ldots,n\}$; horizon $\tau$; number of folds $K$ (e.g., $K=5$).

\medskip
\noindent\textbf{Step 1: Cross-fitting of nuisances.}
Randomly partition the indices $\{1, \dots, n\}$ into $K$ folds $\{\mathcal{I}_k\}_{k=1}^K$. For each fold $k$:
\begin{enumerate}
    \item Train flexible learners on the complement $\mathcal{I}_k^c$ to obtain estimates for the propensity score $\hat\pi^{(-k)}(a\mid w)$, the conditional censoring distribution $\hat G^{(-k)}(t\mid a,w)$, and the conditional failure time distribution $\hat S^{(-k)}(t\mid a,w)$ (and its corresponding cumulative hazard $\widehat\Lambda^{(-k)}(t\mid a,w)$).
    \item For each subject $i\in\mathcal{I}_k$, compute the martingale residuals required for the bias correction:
    \[
    \widehat M^{(-k)}_{i,a}(t) = N_i(t) - \int_0^t Y_i(u)\,d\widehat\Lambda^{(-k)}(u\mid a,W_i).
    \]
   where $N_i(u)=\mathbb{I}\{U_i\le u, \Delta_i=1\}$ and $Y_i(u)=\mathbb{I}\{U_i\ge u\}$.
\end{enumerate}

\medskip

\noindent\textbf{Step 2: Efficient survival curve estimation.}
For each arm $a \in \{0,1\}$ and each time point $t \in [0, \tau]$, construct the one-step estimator $\hat S_a(t)$. We average the bias-corrected estimates across the sample:
\[
\hat S_a(t) = \frac{1}{n}\sum_{k=1}^K \sum_{i \in \mathcal{I}_k} \left\{ \hat S^{(-k)}(t\mid a,W_i) - \frac{\mathbb{I}\{A_i=a\}}{\hat\pi^{(-k)}(a\mid W_i)} \, \int_0^t \frac{ \hat S^{(-k)}(t\mid a,W_i)d\widehat M^{(-k)}_{i,a}(u)}{\hat S^{(-k)}(u\mid a,W_i)\,\hat G^{(-k)}(u\mid a,W_i)} \right\}.
\]
To respect the logical constraint that survival probabilities are non-increasing, we follow \cite{westling2024inference} and subsequently apply monotonic regression (isotonic projection) to the path $\{\hat S_a(t) : t \in [0, \tau]\}$.

\medskip
\noindent\textbf{Step 3: Plug-in AH and log-AH.}
Compute the summary statistics using the efficient survival curves from Step 2:
\[
\hat F_a(\tau) = 1-\hat S_a(\tau), \qquad
\hat R_a(\tau) = \int_0^\tau \hat S_a(u)\,du,
\qquad
\hat\theta_a(\tau) = \log\{\hat F_a(\tau)\}-\log\{\hat R_a(\tau)\}.
\] 
The estimator for the log average hazard ratio is $\hat \beta(\tau) := \hat\theta_1(\tau) - \hat\theta_0(\tau)$.

\medskip
\noindent\textbf{Step 4: Variance estimation.}
By Theorem~\ref{thm:eif-ah}, the influence function for the log-AH in arm $a$ is estimated by:
\[
\widehat\phi^{\,*}_{\theta_a}(O_i)
= -\frac{1}{\hat F_a(\tau)}\,\widehat\phi^{\,*}_{S_a(\tau)}(O_i)
-\frac{1}{\hat R_a(\tau)}\int_0^\tau \widehat\phi^{\,*}_{S_a(u)}(O_i)\,du.
\]
Here, $\widehat\phi^{\,*}_{S_a(t)}(O_i)$ is the plug-in estimate of the survival EIF (Proposition~\ref{prop:eif-mu}) evaluated using the nuisance estimates from folds not containing $i$.
The variance of $\hat{\beta}(\tau)$ is estimated by the empirical variance of the influence functions:
\[
\frac{1}{n}\widehat{\sigma}^2_{\beta} = \frac{1}{n^2} \sum_{i=1}^n \left\{ \widehat\phi^{\,*}_{\theta_1}(O_i)-\widehat{\phi}^{\,*}_{\theta_0}(O_i)-\bar{\phi}_\beta \right\}^2,
\],
where $\bar{\phi}_\beta$ is the empirical average of $\widehat\phi^{\,*}_{\theta_1}(O_i)-\widehat{\phi}^{\,*}_{\theta_0}(O_i).$ The Wald-type confidence interval can be constructed as $\hat\beta(\tau) \pm z_{1-\alpha/2} \widehat{\sigma}_{\beta}/\sqrt{n}.$

\section{Large-Sample Theory for the AH Estimator}
\label{sec:large-sample-theory}

In this section, we establish large-sample properties of the proposed cross-fitted plug-in estimator
of the log average-hazard ratio $\beta(\tau)=\log\{\eta_1(\tau)/\eta_0(\tau)\}=\theta_1(\tau)-\theta_0(\tau).$
Our argument combines: (i) asymptotic theory for cross-fitted one-step estimators of treatment-specific
marginal survival curves under right censoring \citep{westling2024inference}, and
(ii) the functional delta method for the smooth mapping from the survival path to log-AH
(Theorem~\ref{thm:eif-ah}).

\subsection{Assumptions}

We assume the identification conditions (A1)--(A5) from Section~\ref{sec:causal_quantity}.
Let $P_0$ denote the true data distribution and $\mathbb{P}_n$ the empirical distribution.
For measurable $f$, write $\|f\|_{2,P_0}=(\mathbb{E}_{P_0}[f^2])^{1/2}$ and
$\|f\|_\infty=\sup_x |f(x)|$. For functions indexed by $t\in[0,\tau]$, we use the
supremum norm $\|h\|_{\infty,[0,\tau]}=\sup_{t\in[0,\tau]}|h(t)|$.

\medskip

We impose the following conditions for each arm $a\in\{0,1\}$:
\begin{enumerate}
\item[(C1)] \textbf{Positivity and boundedness.}
There exist constant $\epsilon\in(0,1)$  such that almost surely,
\[
\pi_0(a\mid W)\ge \epsilon,\qquad
\inf_{t\in[0,\tau]} G_0(t\mid a,W)\ge \epsilon,\qquad
\inf_{t\in[0,\tau]} S_0(t\mid a,W)\ge \epsilon.
\]
and $ \min\{F_a(\tau), R_a(\tau)\}\ge \epsilon,$
so that inference for $\log \eta_a(\tau)$ is well-defined. In addition, with probability tending to one (uniformly over folds $k$),
\[
\hat\pi^{(-k)}(a\mid W)\ge \epsilon,\qquad
\inf_{t\in[0,\tau]} \hat G^{(-k)}(t\mid a,W)\ge \epsilon,\qquad
\inf_{t\in[0,\tau]} \hat S^{(-k)}(t\mid a,W)\ge \epsilon,
\]
e.g., by truncation of $\hat\pi^{(-k)}$, $\hat G^{(-k)}$, and (if needed) $\hat S^{(-k)}$.

\item[(C2)] \textbf{Nuisance convergence to limits.}
The cross-fitted nuisance estimators converge in probability (uniformly over folds) to fixed limits
$(\pi_\infty, S_\infty, G_\infty)$ in $L_2(P_0)$. Specifically, as the sample size increases:
\begin{align*}
    \max_k\|\hat\pi^{(-k)}(a \mid W) - \pi_\infty(a \mid W)\|_{2,P_0} &\to 0, \\
    \max_k\sup_{t \in [0,\tau]} \|\hat S^{(-k)}(t \mid a,W) - S_\infty(t \mid a,W)\|_{2,P_0} &\to 0, \\
    \max_k\sup_{t \in [0,\tau]} \|\hat G^{(-k)}(t \mid a,W) - G_\infty(t \mid a,W)\|_{2,P_0} &\to 0.
\end{align*}

\item[(C3)] \textbf{Asymptotic linearity of the survival estimator.} 
The cross-fitted one-step estimator $\hat S_a$ (potentially including an isotonic projection to ensure monotonicity) is assumed to admit the uniform expansion:
\begin{equation}
\label{eq:AL_surv_large_sample}
\hat S_a(t) - S_a(t) = \mathbb{P}_n \!\left[ \phi^*_{S_a(t)}(O) \right] + \mathrm{Rem}_{n,a}(t), \qquad \|\mathrm{Rem}_{n,a}\|_{\infty,[0,\tau]} = o_p(n^{-1/2}),
\end{equation}
where $\phi^*_{S_a(t)}$ is the efficient influence function defined in Proposition~\ref{prop:eif-mu}. This expansion is imposed as a high-level condition to transfer inference from $\hat S_a(\cdot)$ to $\hat\beta(\tau)$.
It holds, for example, for the cross-fitted one-step survival estimator under additional regularity conditions;
see \citet[][Thm.~3]{westling2024inference}.

\item[(C4)] \textbf{Product-bias (orthogonality) rate conditions.}
The remainder in \eqref{eq:AL_surv_large_sample} is second order in nuisance errors.
Sufficient conditions are the product rates
\[
\|\hat\pi^{(-k)}(a\mid W)-\pi_0(a\mid W)\|_{2,P_0}\times
\sup_{t\in[0,\tau]}\|\hat S^{(-k)}(t\mid a,W)-S_0(t\mid a,W)\|_{2,P_0}
=o_p(n^{-1/2}),
\]
\[
\sup_{t\in[0,\tau]}\|\hat G^{(-k)}(t\mid a,W)-G_0(t\mid a,W)\|_{2,P_0}\times
\sup_{t\in[0,\tau]}\|\hat S^{(-k)}(t\mid a,W)-S_0(t\mid a,W)\|_{2,P_0}
=o_p(n^{-1/2}),
\]
uniformly over folds $k$.
In particular, these hold if the relevant nuisance components are consistent at rates
faster than $n^{-1/4}$ in $L_2(P_0)$. 
\end{enumerate}

\begin{remark}
Condition (C3) is a high-level uniform asymptotic linear representation for $\hat S_a(\cdot)$ in $\ell^\infty([0,\tau])$,
which is the only input needed to obtain the CLT for $\hat\beta(\tau)$ via the functional delta method.
Condition (C4) is a convenient sufficient product-rate condition that, together with positivity/truncation, controls
second-order remainder terms for one-step estimators. 
\end{remark}



\subsection{Main results}

\begin{theorem}[Consistency and double robustness]
\label{thm:consistency}
Assume (A1)--(A5) and (C1)--(C2). Suppose that for each arm $a\in\{0,1\}$ the nuisance limits satisfy
either
\begin{enumerate}
\item[(i)] $S_\infty(t\mid a,w)=S_0(t\mid a,w)$ for almost all $(t,w)$, or
\item[(ii)] $\pi_\infty(a\mid w)=\pi_0(a\mid w)$ and $G_\infty(t\mid a, w)=G_0(t\mid a, w)$ for almost all $(t,w)$.
\end{enumerate}
Then $\|\hat S_a-S_a\|_{\infty,[0,\tau]}\xrightarrow{P}0$, and consequently
\[
\hat\theta_a(\tau)\xrightarrow{P}\theta_a(\tau),
\qquad
\hat\beta(\tau)\xrightarrow{P}\beta(\tau).
\]
Thus $\hat\beta(\tau)$ is consistent under either nuisance regime (i) or (ii), i.e., it is doubly robust.
\end{theorem}

\begin{theorem}[Asymptotic normality and efficiency]
\label{thm:clt}
Assume the conditions of Theorem~\ref{thm:consistency}. If, in addition, (C3)--(C4) hold for both arms,
then
\[
\sqrt{n}\big\{\hat\beta(\tau)-\beta(\tau)\big\}
=
\frac{1}{\sqrt{n}}\sum_{i=1}^n \phi^*_{\beta}(O_i)
+o_p(1),
\qquad
\phi^*_{\beta}(O)=\phi^*_{\theta_1}(O)-\phi^*_{\theta_0}(O),
\]
where $\phi^*_{\theta_a}(\cdot)$ is the efficient influence function in
Theorem~\ref{thm:eif-ah}. Consequently,
\[
\sqrt{n}\big\{\hat\beta(\tau)-\beta(\tau)\big\}
\ \rightsquigarrow\
\mathcal{N}\!\left(0,\ \mathrm{Var}\{\phi^*_{\beta}(O)\}\right).
\]
Moreover, if the nuisance estimators are consistent for $(S_0,G_0,\pi_0)$,
then $\hat\beta(\tau)$ is asymptotically efficient in the nonparametric model.
\end{theorem}

\begin{corollary}[Inference and Variance Estimation]
\label{cor:variance}
Under the conditions of Theorem~\ref{thm:clt}, a consistent estimator for the asymptotic variance $\sigma^2_\beta = \mathrm{Var}\{\phi^*_\beta(O)\}$ is given by the cross-fitted empirical variance:
\[
\hat\sigma^2_\beta = \frac{1}{n} \sum_{i=1}^n \left( \widehat\phi^{\,*}_{\beta}(O_i) - \bar{\phi}_{\beta} \right)^2,
\]
where $\bar{\phi}_{\beta} = n^{-1}\sum_{i=1}^n \widehat\phi^{\,*}_{\beta}(O_i)$. Wald-type $(1-\alpha)\times 100\%$ confidence intervals for $\beta(\tau)$ constructed as $\hat\beta(\tau) \pm z_{1-\alpha/2} \hat\sigma_\beta/\sqrt{n}$ possess nominal coverage asymptotically.
\end{corollary}

The proofs rely on the Hadamard differentiability of the log-AH mapping $S\mapsto \log\big((1-S(\tau))/\int_0^\tau S\big)$ and the functional delta method; see Appendix~\ref{app:proof-consistency} and Appendix~\ref{app:proof-AN}.

\section{Simulation study}
\label{sec:simulation}

We conducted a Monte Carlo study to evaluate the finite-sample performance of the proposed AH estimators under non-proportional hazards and proportional hazards settings. For each simulation scenario we generated independent datasets with sample sizes
$n \in \{500, 750, 1000, 1250, 1500\}$ and evaluated the log AH ratio
$\beta(\tau) = \log\{\eta_1(\tau)/\eta_0(\tau)\}$ at a follow-up horizon of $\tau = 12$ months.

Across all scenarios, baseline covariates represent age, body-mass index, and a continuous risk score. Treatment is assigned according to a nonlinear propensity score model (Appendix~\ref{app:sim-dgp}), so that simple parametric models for the treatment mechanism are intentionally misspecified. Intercepts in the event and censoring models are calibrated so that, under control, the marginal probability of censoring by 12 months is approximately 20\% and the marginal event probability is about 15\%. Administrative follow-up is truncated at $24$ months in all scenarios.

We focus on two representative data generation mechanisms (DGMs) in the main text: a smooth non-proportional hazards DGM (non-PH and cross-A) and a proportional hazards DGM with either a complex or exponential censoring (PH-complex and PH-simple). Specifically, the PH-complex DGM uses the same Weibull proportional-hazards failure time model and complex censoring distribution from \citet{westling2024inference}, while PH-simple DGM replaces complex censoring distribution with an exponential distribution with a log-rate being a linear function of $(A,W)$ (Appendix~\ref{app:sim-dgp}). The comparison between PH-complex and PH-simple is informative because the censoring distribution can be estimated consistently by our current nuisance library under PH-simple DGM, whereas the censoring distribution is too complex to be consistently estimable under PH-complex DGM; consequently, the practical benefits from double robustness can be attenuated when both outcome and censoring components are poorly captured under PH-complex DGM. Full definitions of all DGMs are provided in Appendix~\ref{app:sim-dgp}.

For each simulated dataset we applied three estimators:
\begin{enumerate}
    \item \texttt{AH-DML}: a cross-fitted, doubly robust estimator of the treatment-specific survival curves and average hazards, implemented via \texttt{CFsurvival} with a stacked library of survival learners;
    \item \texttt{G-computation}: an outcome-regression plug-in estimator based on a survival Super Learner without inverse probability weighting or augmentation;
    \item \texttt{Cox}: a marginal Cox proportional hazards working model with baseline covariate adjustment.
\end{enumerate}
Implementation details and learner libraries are provided in Appendix~\ref{app:sim-methods}. The simulation grid consists of $5$ sample sizes and $1000$ independent replications per configuration.

\begin{table}[t]
\centering
\caption{Whether each nuisance component is consistently estimable by our current library (by DGM).}
\label{tab:dgp-consistency}
\begin{tabular}{lccc}
\toprule
DGM & Outcome & Censoring & Propensity \\
\midrule
non-PH   & yes & no  & yes \\
cross-A  & yes & no  & yes \\
PH-simple  & no  & yes & yes \\
PH-complex & no & no & yes \\
\bottomrule
\end{tabular}
\end{table}

\subsection{Simulation results}

Figure~\ref{fig:main-ratio-grid} summarizes percent bias, variance (rescaled by $100$), MSE (rescaled by $100$), and empirical coverage of nominal $95\%$ intervals for the log average-hazard ratio $\beta(\tau)$ under the non-PH and PH-simple DGMs. Additional results—including the performance of arm-specific AH estimator in more challenging crossing-hazards (cross-A), and the PH-complex censoring settings are reported in Appendix~\ref{app:sim-extra}.  Across sample sizes, \texttt{AH-DML} achieves smaller bias and better empirical coverage than two working-model alternatives, particularly under the non-PH DGM setting, where misspecification of proportional hazards is most consequential for Cox-type inference. The performance of \texttt{AH-DML} in the PH-simple censoring setting is substantially better than that in the PH-complex censoring setting (Appendix~\ref{app:sim-ph}), highlighting the importance of being able to consistently estimate the censoring mechanism.

\begin{figure}[t]
  \centering
  \includegraphics[width=1 \textwidth]{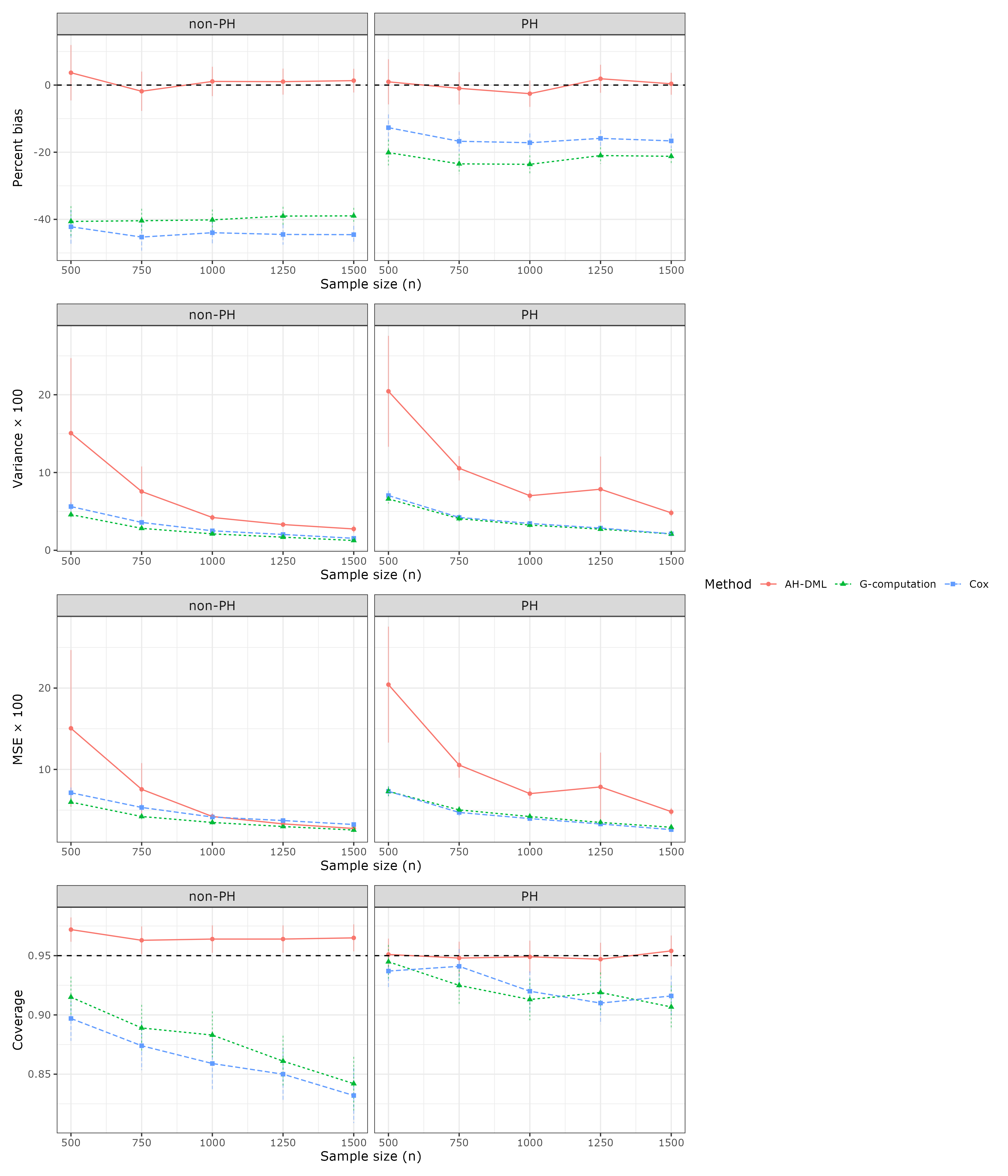}
  \caption{Main simulation summary for the log average-hazard ratio at $\tau=12$ months. Columns correspond to non-PH (left) and PH-simple (right) DGMs. Rows report percent bias, variance $\times 100$, MSE $\times 100$, and empirical coverage of nominal 95\% confidence intervals.}
  \label{fig:main-ratio-grid}
\end{figure}

\section{Data Analysis}
\label{sec:data-analysis}

\subsection{Data Source, Treatment Groups and Outcomes}

We conducted a retrospective cohort study using the Surveillance, Epidemiology, and End Results cancer registry linked with Medicare claims (SEER--Medicare). SEER is a population-based cancer registry program maintained by the National Cancer Institute that collects detailed information on cancer incidence, tumor characteristics, initial treatment, and survival from participating registries covering approximately 28\% of the U.S. population. Medicare is a federally administered health insurance program that provides near-universal coverage for adults aged 65 years and older in the United States and includes comprehensive claims for inpatient, outpatient, and physician services. The linkage of SEER registry data with Medicare enrollment and claims files enables longitudinal, population-based analyses of cancer treatment and outcomes among older adults \citep{warren2002seer, seermedicare2022linkage}.

We used data from the 2022 SEER--Medicare linkage update, which includes cancer diagnoses through 2021 and Medicare claims through 2022. The study cohort consisted of Medicare beneficiaries aged 66 years and older diagnosed with advanced melanoma between 2014 and 2021. To ensure complete capture of baseline comorbidities and treatment exposure, we required continuous enrollment in Medicare Parts A and B for at least 12 months prior to diagnosis and excluded individuals enrolled in Medicare Advantage plans during the baseline or follow-up period. Additional cohort eligibility criteria and variable definitions followed the approved SEER--Medicare data use agreement and analytic specifications.

We focused on patients who initiated one of three first-line systemic immunotherapy regimens within 90 days of melanoma diagnosis: ipilimumab monotherapy, nivolumab monotherapy, or the combination of nivolumab plus ipilimumab. First-line treatment was identified using Medicare procedure and drug claims, and patients whose initial regimen could not be uniquely classified into one of these three categories were excluded.

The primary outcome was overall survival, defined as the time from initiation of first-line immunotherapy to death due to any cause. Patients were censored at the earliest of loss of Medicare fee-for-service coverage or the end of available claims data. Survival time was measured in months.

\subsection{Research Questions and Estimands}

This analysis was designed to replicate, in a real-world population of older adults, the primary treatment comparisons evaluated in the phase 3 CheckMate 067 randomized clinical trial \citep{wolchok2025checkmate}. In that trial, the two prespecified primary hypotheses tested whether (i) nivolumab plus ipilimumab improved overall survival compared with ipilimumab alone and whether (ii) nivolumab monotherapy improved overall survival compared with ipilimumab alone, with treatment effects summarized using Cox proportional hazards models and reported as hazard ratios. A third comparison between nivolumab plus ipilimumab and nivolumab monotherapy was conducted descriptively without formal hypothesis testing.

Accordingly, we examined the following pairwise comparisons in the SEER--Medicare cohort:
\begin{enumerate}
    \item nivolumab plus ipilimumab versus ipilimumab monotherapy;
    \item nivolumab monotherapy versus ipilimumab monotherapy; and
    \item nivolumab plus ipilimumab versus nivolumab monotherapy (descriptive analysis).
\end{enumerate}

Unlike the randomized trial, which summarized treatment effects using hazard ratios, our primary estimand was the average hazard at a prespecified time horizon $\tau$, as defined in Section~\ref{sec:causal_quantity} \footnote{Section~\ref{sec:causal_quantity} puts an emphasis on log AH. In this section, we report the AH which is exponentiated version of log AH. }. For each comparison, we estimated the ratio of average hazards.
In this way, our analysis parallels the clinical questions posed by CheckMate 067 while targeting a distinct, model-free estimand that remains well-defined under non-proportional hazards and delayed treatment effects.

\subsection{Covariate Adjustment and Estimation Details}

Because treatment assignment in SEER--Medicare is not randomized, we adjusted for baseline covariates capturing demographic, clinical, and socioeconomic characteristics that may jointly influence treatment choice and survival. These included age group at diagnosis (66--74, 75--84, $\geq$85 years), sex, race, ethnicity, SEER registry region, marital status, dual Medicare--Medicaid eligibility\footnote{Dual Medicare--Medicaid eligibility indicates concurrent enrollment in Medicare and Medicaid. Dual-eligible beneficiaries are typically low-income individuals for whom Medicaid provides supplemental coverage for Medicare premiums, cost-sharing, and certain services not covered by Medicare.}, census tract--level measures of poverty and educational attainment, Charlson comorbidity index, and the number of hospitalizations in the 90 days prior to treatment initiation. 

Treatment-specific survival curves were estimated using a cross-fitted, doubly robust procedure that combines flexible machine learning models for the outcome, treatment, and censoring mechanisms. These survival curve estimates were then mapped to arm-specific AHs and to the log AH ratio using the efficient influence function derived in Section~\ref{sec:eif-ah}. Wald-type confidence intervals were constructed using the empirical variance of the estimated influence function.

We estimated treatment effects at two prespecified time horizons, $\tau=24$ and $36$ months, to balance clinical relevance with statistical stability of the survival-weighted estimand. Because the AH depends on $\tau$, reliable estimation requires adequate follow-up with sufficient numbers at risk near $\tau$. In the SEER--Medicare cohort, the number of individuals with observed survival time at least 24 months was 85 (nivolumab plus ipilimumab), 86 (nivolumab), and 20 (ipilimumab), whereas at 36 months the corresponding counts were reduced to 56, 56, and 15, respectively. Those counts decreased further beyond 36 months. We therefore report results at both 24 and 36 months to assess the sensitivity of AH-based inference to the choice of time horizon.

\subsection{Results}
\label{subsec:data-analysis-results}

Both nivolumab plus ipilimumab and nivolumab monotherapy were associated with substantially lower average hazards of death than ipilimumab monotherapy at both prespecified horizons. At 24 months, the estimated AH ratio was 0.60 (95\% CI, 0.39--0.90) for nivolumab plus ipilimumab versus ipilimumab and 0.59 (95\% CI, 0.39--0.89) for nivolumab versus ipilimumab. Corresponding estimates at 36 months were similar, with AH ratios of 0.59 (95\% CI, 0.39--0.90) and 0.56 (95\% CI, 0.38--0.84), respectively. These results indicate a stable reduction in the AH for both nivolumab-containing regimens relative to ipilimumab monotherapy across time horizons.

\paragraph{Secondary comparison.}
For the descriptive comparison between nivolumab plus ipilimumab and nivolumab monotherapy, estimated AH ratios were close to the null at both horizons. The AH ratio was 1.01 (95\% CI, 0.74--1.38) at 24 months and 1.13 (95\% CI, 0.78--1.63) at 36 months. These findings suggest no or small difference in AH between combination therapy and nivolumab monotherapy in this population. 

\paragraph{Sensitivity to horizon.}
Across comparisons, estimates at 24 months were consistently more precise than those at 36 months, reflecting the smaller numbers of patients remaining at risk beyond three years of follow-up. Nevertheless, point estimates were qualitatively similar across horizons, indicating that conclusions regarding the relative effectiveness of nivolumab-containing regimens versus ipilimumab monotherapy were robust to the choice of time horizon within considered range. These findings are consistent with the expected trade-off between clinical relevance and statistical stability when estimating survival-weighted causal estimands at longer horizons.

Figure~\ref{fig:seer-ah} summarizes the estimated log average-hazard ratios and confidence intervals across horizons, analysis settings, and treatment comparisons.

\begin{figure}[t]
  \centering
  \includegraphics[width=0.9\textwidth]{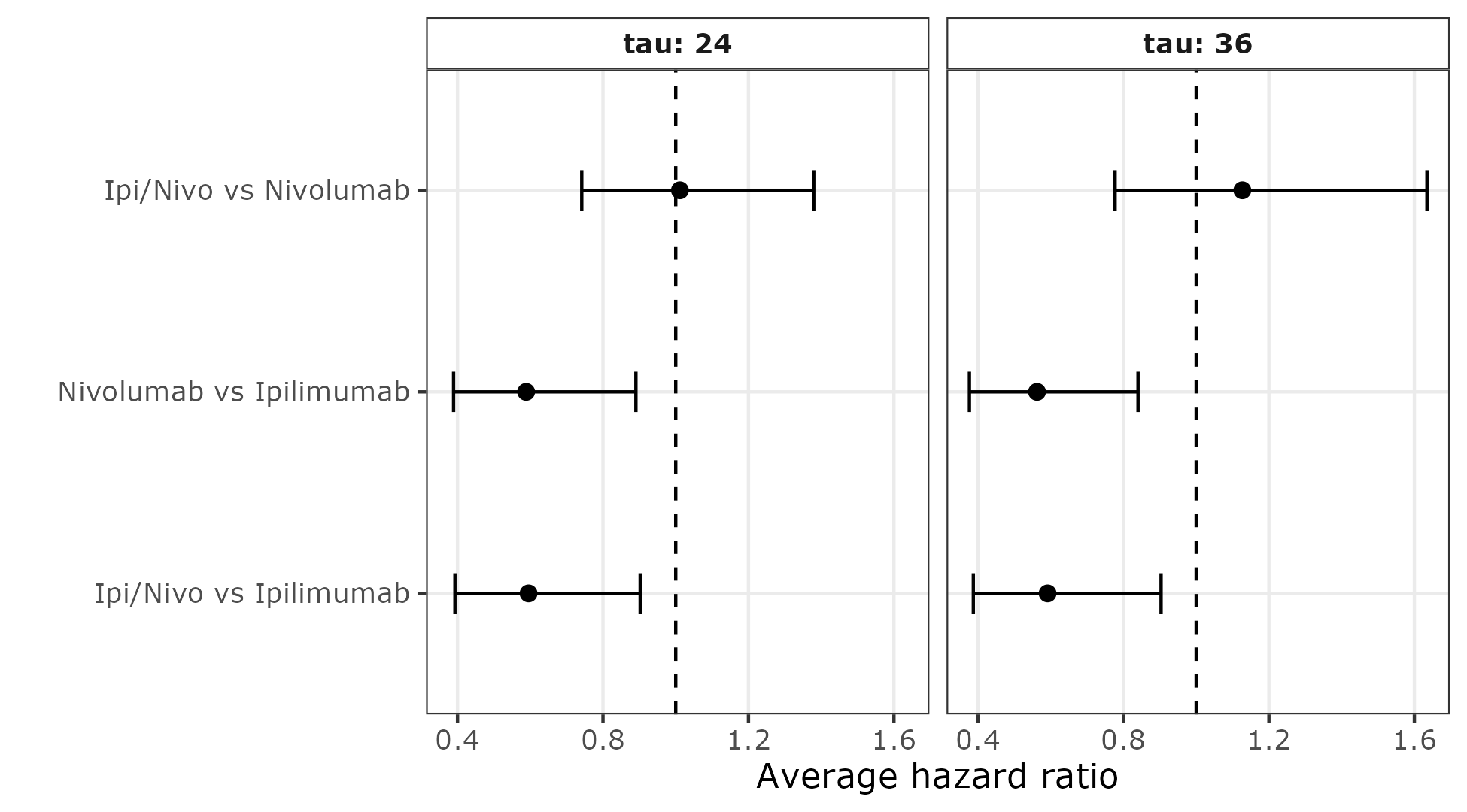}
  \caption{Estimated average hazard (AH) ratios  with 95\% Wald confidence intervals for comparisons of nivolumab plus ipilimumab, nivolumab monotherapy, and ipilimumab monotherapy at 24- and 36-month horizons. The vertical dashed line indicates no difference in average hazard.}
  \label{fig:seer-ah}
\end{figure}

\section{Conclusion}


We studied semiparametric estimation and inference for the AH, a rate-based causal estimand that remains well-defined under non-proportional hazards. By representing the estimand as a smooth functional of the marginal survival function, we derived the efficient influence function and constructed cross-fitted debiased estimators with doubly robust property.

First, we formalized the AH as a causal estimand within the potential outcomes framework. We demonstrated that defining the AH via marginal survival functions—representing a population-level person-time event rate—is the appropriate approach for causal inference. This avoids the non-collapsibility issues inherent in simply averaging conditional incidence rates, which can very different from the true population-level treatment effect.

Second, we established the pathwise differentiability of the AH and derived its efficient influence function. By viewing the AH as a smooth functional of the survival path, we utilized the functional delta method to provide a foundation for efficient inference. This allows for the construction of cross-fitted, doubly robust estimators that can leverage high-dimensional nuisance parameter estimation via machine learning while maintaining consistency and asymptotic normality.

Our simulation results confirm the theoretical advantages of this approach in practice. The proposed AH-DML estimator maintains minimal bias across diverse settings, including challenging scenarios with crossing hazards where traditional Cox-based hazard ratios and simple G-computation fail to provide reliable results. Furthermore, our application to SEER-Medicare data demonstrates the utility of the AH in real-world comparative effectiveness research. The AH ratios for immunotherapy regimens remained stable across two time horizons, providing a robust and clinically intuitive summary of treatment effects in US patient population.

There are several promising avenues for future research. While our framework focuses on binary treatments and right-censored data, extending the efficient influence function derivations to accommodate multiple treatment arms, continuous exposures, or competing risks would further broaden the applicability of the AH. Additionally, investigating the performance of targeted maximum likelihood estimation specifically tailored to the AH functional could offer further gains in finite-sample robustness. By integrating modern semiparametric theory with flexible machine learning, the AH provides a principled and powerful tool for the analysis of complex time-to-event data.

\section{Acknowledgement}
This research was supported by the National Institute of General Medical Sciences of the National Institutes of Health under award number R01 GM152499, the National Heart Lung and Blood Institute of the National Institutes of Health under award number R01 HL089778, and the McGraw/Patterson Research Fund. 

We thank Morgan A. Paul, MB, for her invaluable assistance in preparing the SEER–Medicare dataset for the analysis. We also acknowledge the expertise and guidance of Angela C. Tramontano regarding the SEER-Medicare data access and structure. 

The collection of cancer incidence data used in this study was supported by the California Department of Public Health pursuant to California Health and Safety Code Section 103885; Centers for Disease Control and Prevention’s (CDC) National Program of Cancer Registries, under cooperative agreement NU58DP007156; the National Cancer Institute’s Surveillance, Epidemiology and End Results Program under contract HHSN261201800032I awarded to the University of California, San Francisco, contract HHSN261201800015I awarded to the University of Southern California, and contract HHSN261201800009I awarded to the Public Health Institute. The ideas and opinions expressed herein are those of the author(s) and do not necessarily reflect the opinions of the State of California, Department of Public Health, the National Cancer Institute, and the Centers for Disease Control and Prevention or their Contractors and Subcontractors.

\bibliographystyle{apalike}
\bibliography{refs}

\appendix

\section{Proofs}
\label{app:proofs}

\subsection{Proof of Theorem~\ref{thm:identification} and Proposition~\ref{prop:identification-censoring}}
\label{app:proof-identification}

We prove that under assumptions (A1)--(A5), the causal AH
\[
\eta_a(\tau) = \frac{\Pr\{T(a) \le \tau\}}{\mathbb{E}[\min\{T(a),\tau\}]}
\]
is identified from the observed data distribution by
\[
\eta_a(\tau) = \frac{1 - S_a(\tau)}{\int_0^\tau S_a(t)\,dt},
\qquad
S_a(t) = \mathbb{E}\big[ S_0(t \mid a,W) \big].
\]

\paragraph{Step 1. Identification of survival under intervention.}
By consistency (A1), if $A=a$, then $T=T(a)$. By conditional exchangeability (A2),
\[
\Pr\{T(a) > t\}
= \mathbb{E}\big[ \Pr\{T(a) > t \mid W\} \big]
= \mathbb{E}\big[ \Pr\{T > t \mid A=a,W\} \big].
\]
The last expression equals $S_a(t) = \mathbb{E}[S_0(t \mid a,W)]$, which is a functional of the observed data distribution.

\paragraph{Step 2. Identification of the numerator.}
The counterfactual cumulative incidence at horizon $\tau$ is
\[
F_a(\tau) = \Pr\{T(a) \le \tau\} = 1 - \Pr\{T(a) > \tau\} = 1 - S_a(\tau).
\]

\paragraph{Step 3. Identification of the denominator.}
Note that
\[
R_a(\tau) = \mathbb{E}[\min\{T(a),\tau\}]
= \int_0^\tau \Pr\{T(a) > t\}\,dt.
\]
By Step 1, $\Pr\{T(a) > t\} = S_a(t)$, hence
\[
R_a(\tau) = \int_0^\tau S_a(t)\,dt.
\]

\paragraph{Step 4. Ratio representation.}
Combining Steps 2 and 3, we obtain
\[
\eta_a(\tau) = \frac{F_a(\tau)}{R_a(\tau)}
= \frac{1 - S_a(\tau)}{\int_0^\tau S_a(t)\,dt}.
\]

\paragraph{Step 5. Mapping to the observed-data law under censoring.}
Under (A4)--(A5), the conditional survival $S_0(t\mid a,W)$ is identified from the observed-data distribution of
$(U,\Delta)$ given $(A=a,W)$ via the standard product-integral (equivalently, conditional hazard) representation.
Assumption (A3) ensures $\pi_a(W)$ is bounded away from 0 and 1, and (A5) ensures the relevant censoring
survival probabilities are strictly positive on $[0,\tau]$, so the identified functionals are well-defined.

Thus, the causal AH is identified as claimed.
\qed

\subsection{Proof of Theorem~\ref{thm:eif-ah}}
\label{appendix:eif-ah}

We briefly recall the notions of pathwise differentiability and Hadamard
differentiability used in the proof; see \citet[][Chs.~20, 25]{vaart1998asymptotic}
and \citet[][Ch.~3]{bickel1993efficient}.

\paragraph{Pathwise differentiability.}
Let $\mathcal{P}$ be a (possibly infinite-dimensional) model for the data law $P$,
and let $\Psi:\mathcal{P}\to\mathcal{H}$ be a parameter taking values in a Banach
space $\mathcal{H}$ (for us, typically $\mathcal{H}=\ell^\infty([0,\tau])$).
We say that $\Psi$ is \emph{pathwise differentiable} at $P_0\in\mathcal{P}$ if, for
every regular one-dimensional submodel
$\{P_\varepsilon:\varepsilon\in(-1,1)\}\subset\mathcal{P}$ passing through $P_0$
with score $s\in L_2^0(P_0)$, the map $\varepsilon\mapsto\Psi(P_\varepsilon)$ is
differentiable at $\varepsilon=0$ and there exists
$\phi_\Psi\in L_2^0(P_0;\mathcal{H})$ such that
\[
\frac{\mathrm{d}}{\mathrm{d}\varepsilon}\Big|_{\varepsilon=0}\Psi(P_\varepsilon)
=
\mathbb{E}_{P_0}\!\big[\phi_\Psi(O)\,s(O)\big].
\]
The element $\phi_\Psi$ is the (canonical) efficient influence function of $\Psi$ at $P_0$.

\paragraph{Hadamard differentiability.}
Let $\Phi:\mathcal{D}\subset\mathcal{H}\to\mathbb{R}$ be a functional defined on a
subset $\mathcal{D}$ of a Banach space $\mathcal{H}$.  We say that $\Phi$ is
Hadamard differentiable at $\eta\in\mathcal{D}$ tangentially to a subspace
$\mathcal{D}_0\subset\mathcal{H}$ if there exists a continuous linear map
$\dot\Phi_\eta:\mathcal{D}_0\to\mathbb{R}$ such that, for every sequence
$h_n\to h$ in $\mathcal{D}_0$ and every $\varepsilon_n\downarrow0$ with
$\eta+\varepsilon_n h_n\in\mathcal{D}$,
\[
\frac{\Phi(\eta+\varepsilon_n h_n)-\Phi(\eta)}{\varepsilon_n}
\;\longrightarrow\;
\dot\Phi_\eta[h].
\]
Hadamard differentiability is the standard regularity condition supporting the
functional delta method in Banach spaces
\citep[][Defs.~20.1–20.2, Thm~20.8]{vaart1998asymptotic}.

\paragraph{Functional delta method and influence-function chain rule.}
Let $\Psi:\mathcal{P}\to\mathcal{H}$ be pathwise differentiable at $P_0$ with
efficient influence function $\phi_\Psi$, and let $\Phi:\mathcal{H}\to\mathbb{R}$
be Hadamard differentiable at $\Psi(P_0)$ tangentially to the range of the tangent
space of $\Psi$.  Consider the composition
\[
\Theta(P) \;:=\; \Phi\big(\Psi(P)\big).
\]
Then $\Theta$ is pathwise differentiable at $P_0$, and the functional delta method
\citep[][Thm~20.8]{vaart1998asymptotic}, as formulated in influence-function
language in \citet[][Sec.~3.3]{bickel1993efficient}, yields
\begin{equation}
\label{eq:if-chain-rule-app}
\phi^*_\Theta(O)
\;=\;
\dot\Phi_{\Psi(P_0)}\!\big[\phi_\Psi(O)\big],
\end{equation}
that is, the EIF of $\Theta$ is the image of the EIF of $\Psi$ under the
Hadamard derivative of $\Phi$.

\medskip

We now apply this general result to the log average hazard.

\begin{proof}[Proof of Theorem~\ref{thm:eif-ah}]
Fix an arm $a\in\{0,1\}$ and write
\[
S_a(t)=\mathbb{E}\big[S_0(t\mid a,W)\big],\qquad
F_a(\tau)=1-S_a(\tau),\qquad
R_a(\tau)=\int_0^\tau S_a(u)\,du.
\]
Define the functional $\Phi$ on a suitable subset
$\mathcal{D}\subset L^\infty([0,\tau])$ by
\[
\Phi(S)
\;=\;
\log\!\left(\frac{1-S(\tau)}{\int_0^\tau S(u)\,du}\right)
= \log\!\big(F(\tau)/R(\tau)\big),
\]
where $F(\tau)=1-S(\tau)$ and $R(\tau)=\int_0^\tau S(u)\,du$. We view
$\theta_a$ as the composition
\[
\theta_a \;=\; \Theta(P_0) := \Phi\big(\Psi(P_0)\big),
\qquad \Psi(P)(t) = S_a(t).
\]

\medskip\noindent
\textbf{Step 1: Hadamard derivative of $\Phi$.}
Let $S\in\mathcal{D}$ be such that $F(\tau)>0$ and $R(\tau)>0$, and let
$h\in L^\infty([0,\tau])$ be a bounded perturbation. For small $\varepsilon$,
set $S_\varepsilon=S+\varepsilon h$ and corresponding
\[
F_\varepsilon(\tau)=1-S_\varepsilon(\tau),
\qquad
R_\varepsilon(\tau)=\int_0^\tau S_\varepsilon(u)\,du.
\]
Differentiating gives
\[
F'(\tau)[h]
=\frac{\mathrm{d}}{\mathrm{d}\varepsilon}\Big|_{0} F_\varepsilon(\tau)
=-h(\tau),
\qquad
R'(\tau)[h]
=\frac{\mathrm{d}}{\mathrm{d}\varepsilon}\Big|_{0} R_\varepsilon(\tau)
=\int_0^\tau h(u)\,du.
\]
Using the chain rule for the smooth map $(a,b)\mapsto \log(a/b)$ we obtain
\[
\dot\Phi_S[h]
=
\frac{F'(\tau)[h]}{F(\tau)}-\frac{R'(\tau)[h]}{R(\tau)}
=
-\frac{h(\tau)}{F(\tau)}-\frac{1}{R(\tau)}\int_0^\tau h(u)\,du.
\]
The map $h\mapsto\dot\Phi_S[h]$ is linear and continuous on
$L^\infty([0,\tau])$, so $\Phi$ is Hadamard differentiable at $S$
tangentially to $L^\infty([0,\tau])$ whenever $F(\tau),R(\tau)>0$;
this is an instance of \citet[][Thm~20.8]{vaart1998asymptotic}.

\medskip\noindent
\textbf{Step 2: Pathwise differentiability of $S_a(\cdot)$.}
By Proposition~\ref{prop:eif-mu} (adapted from \citealp{westling2024inference}),
under (A1)--(A5) the map $\Psi:\mathcal{P}\to\ell^\infty([0,\tau])$ given by
$\Psi(P)(t)=S_a(t)$ is pathwise differentiable at $P_0$ in the nonparametric
model, with efficient influence function process $t\mapsto\phi^*_{S_a(t)}(O)$.

\medskip\noindent
\textbf{Step 3: Chain rule for the log AH parameter.}
Consider the scalar parameter
\[
\Theta(P) := \Phi\big(\Psi(P)\big)
= \log\!\left(
\frac{1-\Psi(P)(\tau)}
     {\int_0^\tau \Psi(P)(u)\,du}
\right),
\]
so that $\Theta(P_0)=\theta_a$.
Since $\Psi$ is pathwise differentiable at $P_0$ and $\Phi$ is Hadamard
differentiable at $\Psi(P_0)=S_a$, the composition $\Theta$ is pathwise
differentiable at $P_0$ and its EIF is given by the chain rule
\eqref{eq:if-chain-rule-app}:
\[
\phi^*_{\theta_a}(O)
=
\dot\Phi_{S_a}\!\big[\phi^*_{S_a(\cdot)}(O)\big].
\]
Substituting the expression for $\dot\Phi_{S_a}$ obtained in Step~1 with
$h(\cdot)=\phi^*_{S_a(\cdot)}(O)$ yields
\[
\phi^*_{\theta_a}(O)
= -\,\frac{1}{F_a(\tau)}\,\phi^*_{S_a(\tau)}(O)
\;-\;\frac{1}{R_a(\tau)}\int_0^\tau \phi^*_{S_a(u)}(O)\,du,
\]
which is exactly \eqref{eq:eif-ah-from-mu}. This completes the proof.
\end{proof}

\subsection{Proof of Theorem~\ref{thm:consistency}}
\label{app:proof-consistency}

\begin{proof}
Fix $a\in\{0,1\}$. The cross-fitted DML estimator for the marginal survival curve is defined as:
\[
\hat S_a(t) = \frac{1}{n} \sum_{k=1}^K \sum_{i \in \mathcal{I}_k} \phi_{t,a}(O_i, \hat{\nu}^{(-k)}),
\]
where $\hat{\nu}^{(-k)} = (\hat{S}^{(-k)}, \hat{G}^{(-k)}, \hat{\pi}^{(-k)})$ is the set of nuisance parameters estimated from the training data excluding fold $k$, and $\phi_{t,a}$ is the influence function defined in Proposition~\ref{prop:eif-mu}.

We decompose the error $\hat S_a(t)-S_a(t)$ into an empirical-process term and a bias (drift) term:
\begin{equation}
\label{eq:AL_decomposition}
\hat S_a(t)-S_a(t)
=
(\mathbb{P}_n-\mathbb{P}_0)\phi_{t,a}(O, \hat{\nu})
\;+\;
\Big\{P_0 \phi_{t,a}(O, \hat{\nu}) - S_a(t)\Big\},
\end{equation}
where $P_0 \phi_{t,a}(O, \hat{\nu})$ denotes the expectation over the data $O$ holding the nuisance estimates $\hat{\nu}$ fixed (standard under the cross-fitting paradigm).

Under cross-fitting, for each fold $k$ the function class
$\{\phi_{t,a}(\cdot,\hat\nu^{(-k)}): t\in[0,\tau]\}$
is evaluated on observations independent of the training sample used to construct $\hat\nu^{(-k)}$.
Combined with the boundedness/positivity in (C1) and the convergence in (C2), standard arguments for cross-fitted
one-step survival estimators yield
\[
\sup_{t\in[0,\tau]}\left|(\mathbb{P}_n-\mathbb{P}_0)\phi_{t,a}(O,\hat\nu)\right|=o_p(1),
\]
see, e.g., \citet[][proof of Thm.~2]{westling2024inference} for a detailed treatment of the uniform control.


For the bias term, we utilize the drift representation for survival functionals \citep{westling2024inference}:
\begin{equation}
\label{eq:drift_app}
\begin{split}
P_0\phi_{t,a}(O, \hat{\nu}) - S_a(t)
= \mathbb{E}\Bigg[
\hat S(t\mid a,W) \int_0^{t} 
&\frac{S_0(u^{-}\mid a,W)}{\hat S(u\mid a,W)}
\left\{\frac{\pi_0(a\mid W)\,G_0(u\mid a,W)}{\hat \pi(a\mid W)\,\hat G(u\mid a,W)}-1\right\} \\
&\times d\!\left\{\hat \Lambda(u\mid a,W)-\Lambda_0(u\mid a,W)\right\}
\Bigg].
\end{split}
\end{equation}

By Neyman orthogonality, this term consists of products of nuisance estimation errors. Specifically, the drift vanishes if either $\hat S$ is consistent for $S_0$ (making the $d\{\hat \Lambda - \Lambda_0\}$ term zero in the limit) or if the pair $(\hat \pi, \hat G)$ is consistent for $(\pi_0, G_0)$ (making the ratio term zero). These correspond exactly to the double-robustness conditions (i) and (ii) in Theorem~\ref{thm:consistency}. Given these conditions, we have $\|\hat S_a-S_a\|_{\infty,[0,\tau]} \xrightarrow{P} 0$.

To establish consistency for the final estimand, we invoke the Continuous Mapping Theorem. Let $\Phi: \ell^\infty([0, \tau]) \to \mathbb{R}^2$ be the map $S \mapsto (1-S(\tau), \int_0^\tau S(u)du)$. Since $\Phi$ is continuous with respect to the supremum norm, $(\hat F_a(\tau), \hat R_a(\tau)) \xrightarrow{P} (F_a(\tau), R_a(\tau))$. Furthermore, the ratio map $g(x, y) = x/y$ is continuous on the domain where $y \ge \epsilon > 0$, which is guaranteed by (C1). Thus, $\hat\eta_a(\tau) \xrightarrow{P} \eta_a(\tau)$. Finally, the log-transformation $\beta(\tau) = \log(\eta_1(\tau)) - \log(\eta_0(\tau))$ is continuous on $(0, \infty)$, and (C1) ensures $F_a(\tau)$ is bounded away from zero, completing the proof.
\end{proof}

\subsection{Proof of Theorem~\ref{thm:clt}}
\label{app:proof-AN}

By condition (C3), for each arm $a$ we have the uniform asymptotic linear expansion
\begin{equation}
\label{eq:surv_AL_app}
\sup_{t\in[0,\tau]}
\left|
\hat S_a(t)-S_a(t)-\mathbb{P}_n\!\left[\phi^*_{S_a(t)}(O)\right]
\right|
=o_p(n^{-1/2}),
\end{equation}
where $\phi^*_{S_a(t)}$ is the efficient influence function in Proposition~\ref{prop:eif-mu}.

Define the map $\Gamma:\ell^\infty([0,\tau])^2\to\mathbb{R}$ by
\[
\Gamma(f_0,f_1)
=
\log\!\left(\frac{1-f_1(\tau)}{\int_0^\tau f_1(u)\,du}\right)
-\log\!\left(\frac{1-f_0(\tau)}{\int_0^\tau f_0(u)\,du}\right),
\]
so that $\beta(\tau)=\Gamma(S_0(\cdot),S_1(\cdot))$ and $\hat\beta(\tau)=\Gamma(\hat S_0(\cdot),\hat S_1(\cdot))$.

By (C1), both $F_a(\tau)$ and $R_a(\tau)$ are bounded away from zero, and the one-dimensional map
$f\mapsto \log((1-f(\tau))/\int_0^\tau f)$ is Hadamard differentiable at $S_a$ (as shown in Appendix~\ref{appendix:eif-ah}).
Hence $\Gamma$ is Hadamard differentiable at $(S_0,S_1)$, and the functional delta method applied to
\eqref{eq:surv_AL_app} (jointly in $a=0,1$) gives
\[
n^{1/2}\{\hat\beta(\tau)-\beta(\tau)\}
=
\frac{1}{\sqrt{n}}\sum_{i=1}^n \phi^*_{\beta}(O_i) + o_p(1),
\]
where $\phi^*_{\beta}=\phi^*_{\theta_1}-\phi^*_{\theta_0}$ and $\phi^*_{\theta_a}$ is the mapped EIF in
\eqref{eq:eif-ah-from-mu} (Theorem~\ref{thm:eif-ah}). The CLT follows from the classical i.i.d.\ central limit theorem
since $\phi^*_{\beta}\in L_2^0(P_0)$.

Finally, because $\widehat\phi^{\,*}_{\beta}$ is a cross-fitted plug-in estimator of $\phi^*_{\beta}$,
$\hat\sigma^2_\beta=\mathbb{P}_n[(\widehat\phi^{\,*}_{\beta}-\mathbb{P}_n\widehat\phi^{\,*}_{\beta})^2]$ is consistent for
$\mathrm{Var}\{\phi^*_{\beta}(O)\}$, and thus $\widehat{\mathrm{Var}}\{\hat\beta(\tau)\}=\hat\sigma^2_\beta/n$ is consistent.
\qed

\section{Simulation data-generating mechanisms}
\label{app:sim-dgp}

This appendix describes the data-generating mechanisms (DGMs) used in the simulation study of Section~\ref{sec:simulation}. All DGMs share the same baseline covariates and treatment mechanism; they differ only in the conditional event-time distribution and the censoring model.

\subsection{Baseline covariates and treatment}

For each subject we generate a three-dimensional baseline covariate vector
$W = (W_1,W_2,W_3)$ representing age, BMI, and a risk score. Let $\mathrm{Beta}(a,b)$ denote a beta distribution on $(0,1)$. We set
\begin{align*}
W_1 &\sim 20 + 60 \cdot \mathrm{Beta}(1.1,1.1), \\
W_3 \mid W_1 &= 10 \cdot \mathrm{Beta}\!\Bigl(1.5 + \frac{|W_1 - 50|}{20},\,3\Bigr),\\
W_2 \mid W_1 &= 18 + 32 \cdot \mathrm{Beta}\!\Bigl(1.5 + \frac{W_1}{20},\,6\Bigr).
\end{align*}
This yields ages approximately in $[20,80]$, BMI positively associated with age, and risk scores highest among very young and very old individuals.

Treatment $A \in \{0,1\}$ is assigned via the nonlinear propensity score
\[
\Pr(A = 1 \mid W = w)
= \mathrm{expit}\!\left( -1 + \log\Bigl\{1 + \exp\bigl(-20 + w_1 / 10\bigr) + \exp\bigl(-3 + w_3 / 2\bigr)\Bigr\} \right),
\]
and we draw $A \mid W \sim \mathrm{Bernoulli}\{\Pr(A=1 \mid W)\}$. This design ensures positivity while deliberately misspecifying simple parametric logistic models.

In all DGMs we define the observed follow-up time and event indicator as
\[
U = \min\{T(A), C\}, \qquad \Delta = \mathbb{I}\{T(A) \le C\},
\]
and apply administrative truncation at $24$ months.

\subsection{Proportional hazards with complex censoring (PH-complex)}
\label{app:sim-ph-complex-dgp}

The PH-complex DGM is based on a Weibull model on the hazard scale with Cox-type linear predictors in $(A,W)$, adapted from \citet{westling2024inference}. Let $T$ denote event time and $C$ censoring time. Conditional on $(A,W)$:
\begin{align*}
h_T(t \mid A = a, W = w)
&= t^{\alpha - 1} \exp\{\eta_T(a, w)\},\\
h_C(t \mid A = a, W = w)
&= t^{\alpha - 1} \exp\{\eta_C(a, w)\},
\end{align*}
where $\alpha > 0$ is a shape parameter (we use $\alpha = 1.5$) and
\begin{align*}
\eta_T(a, w)
&= S_0
+ (S_1 - S_0)\,a
+ 0.2 \frac{w_1 - 50}{10}
- 0.15 a \frac{w_1 - 50}{10}
+ 0.2 (w_3 - 5)
- 0.15 a (w_3 - 5) \\
&\quad
+ 0.1 (w_3 - 5)\frac{w_1 - 50}{10}
+ 0.05 (w_2 - 30)
- 0.025 a (w_2 - 30),
\\[0.5em]
\eta_C(a, w)
&= G_0
- 0.4 a
+ 0.2 \frac{w_1 - 50}{10}
- 0.15 a \frac{w_1 - 50}{10}
+ 0.2 (w_3 - 5)
- 0.15 a (w_3 - 5) \\
&\quad
+ 0.1 (w_3 - 5)\frac{w_1 - 50}{10}
+ 0.05 (w_2 - 30)
- 0.025 a (w_2 - 30),
\end{align*}
with $S_0=-5.6$, $S_1=-5.96$, $G_0=-4.7$.

In this DGM, the hazard ratio comparing treatment $A=1$ to control $A=0$ is
\[
\frac{h_T(t \mid 1, w)}{h_T(t \mid 0, w)}
= \exp\{\eta_T(1,w) - \eta_T(0,w)\},
\]
which does not depend on time $t$, so proportional hazards holds for fixed $W$.

\paragraph{Why ``PH-complex''?}
Both the event and censoring predictors include interaction structure that is not consistently estimable by our current nuisance library, so double robustness need not protect the AH-DML estimator in this regime.

\subsection{Proportional hazards with simple exponential censoring (PH-simple)}
\label{app:sim-ph-linear-dgp}

We also consider a PH variant that keeps the same Weibull PH event-time model as PH-complex,
\[
h_T(t \mid A=a, W=w) = t^{\alpha-1}\exp\{\eta_T(a,w)\},
\]
with $\alpha=1.5$, $S_0=-5.6$, and $S_1=-5.96$, but uses a simpler exponential censoring distribution with a linear log-rate in $(A,W)$:
\[
C \mid A=a, W=w \sim \mathrm{Exp}\bigl(\exp\{\eta_C^{\text{lin}}(a,w)\}\bigr),
\]
where
\[
\eta_C^{\text{lin}}(a,w)
= \beta_{0,c}
+ \beta_{A,c}\,a
+ \beta_{1,c}\frac{w_1-50}{10}
+ \beta_{2,c}\frac{w_2-30}{10}
+ \beta_{3,c}(w_3-5).
\]
We fix
\[
(\beta_{A,c},\beta_{1,c},\beta_{2,c},\beta_{3,c}) = (0.3,\,0.20,\,0.05,\,0.20),
\qquad \beta_{0,c}=-3.8,
\]
yielding approximately
$E[P(C \le 12 \mid A=0,W)] \approx 0.20$ and
$E[P(T \le C \mid A=0,W)] \approx 0.15$.
Censoring is administratively truncated at 24 months.

\paragraph{Treatment mechanism.}
Treatment assignment uses the same nonlinear propensity score as above.

\subsection{Scenario: non-proportional hazards (non-PH)}
\label{app:sim-nonph-dgp}

For the warp-based DGMs, define a baseline exponential rate
\[
\lambda_0(w)
= \exp\!\left\{ \beta_0
  - \frac{|w_1 - 60|}{10}
  - 2 \log (w_2)
  + \frac{w_3}{2} \right\}, \qquad \beta_0 = 1.0,
\]
so that under control $T(0) \mid W=w \sim \mathrm{Exp}\{\lambda_0(w)\}$, together with covariate-dependent warp parameters
\begin{align*}
\gamma(w)
&= \mathrm{expit}\!\Biggl(
 \beta_\gamma
 + \frac{1}{2} \log\bigl(1 + e^{(w_1 - 55)/5}\bigr)
 + \frac{1}{4} \log\bigl(1 + e^{(w_2 - 30)/3}\bigr)
 \Biggr),
 \qquad \beta_\gamma = -1.64,\\[0.3em]
\iota(w)
&= \exp\!\Biggl(
 2
 - \frac{1}{2} \log\bigl(1 + e^{(w_1 - 55)/5}\bigr)
 - \frac{1}{10} \log\bigl(1 + e^{(w_2 - 30)/3}\bigr)
 \Biggr),
\end{align*}
where $\gamma(w)$ controls the maximal hazard reduction under treatment and $\iota(w)$ its duration.

In the non-PH warp DGM, control event times satisfy
\[
T(0) \mid W=w \sim \mathrm{Exp}\{\lambda_0(w)\},
\qquad
S_0(t \mid A=0,W=w) = \exp\{-\lambda_0(w)t\}.
\]
For treated subjects we define a covariate-dependent time warp $\phi(t,w)$ and set
\[
S(t \mid A=1,W=w) = S_0\bigl(\phi(t,w) \mid W=w\bigr)
= \exp\{-\lambda_0(w)\,\phi(t,w)\}.
\]
We take
\[
\phi'_{\text{nonPH}}(t,w) =
\begin{cases}
1 - \dfrac{t}{r}\bigl(1 - \gamma(w)\bigr), & 0 \le t \le r,\\[0.5em]
\gamma(w), & r < t \le r + \iota(w),\\[0.5em]
1 - \dfrac{(1 - \gamma(w))\{r + \iota(w)\}^2}{t^2}, & t > r + \iota(w),
\end{cases}
\]
with ramp-up length $r = 1.5$ months. Censoring times are generated from an exponential distribution with covariate- and treatment-dependent rate
\[
C \mid (A=a,W=w) \sim \mathrm{Exp}\Bigl(\exp\{\beta_1 + 0.3 a + \log(1 + \exp{(30 - w_1)/3}) + ((w_3-5)/4)^2\}\Bigr),
\]
with $\beta_1 = -4.9$, and administratively truncated at 24 months.

\subsection{Scenario: crossing hazards (cross-A) DGM}
\label{app:sim-cross-dgp}

The crossing-hazards DGM modifies the warp to produce early benefit and late harm. We again set $T(0) \mid W=w \sim \mathrm{Exp}\{\lambda_0(w)\}$ and define the treated hazard via
\[
h_{\text{crossA}}(t \mid A=1,w) = \lambda_0(w)\,\phi'_{\text{crossA}}(t,w),
\]
where
\[
\phi'_{\text{crossA}}(t,w)
=
\begin{cases}
1 - \gamma(w), & 0 \le t \le r,\\[0.5em]
\gamma(w),     & r < t \le r + \iota(w),\\[0.5em]
1 + \varepsilon\bigl(1 - \gamma(w)\bigr), & t > r + \iota(w),
\end{cases}
\]
with $r = 1.5$ months and $\varepsilon = 0.5$.
Censoring is generated as in the non-PH warp DGM.

\section{Simulation estimators and implementation}
\label{app:sim-methods}

This appendix provides additional details on the three estimators compared in the simulation study of Section~\ref{sec:simulation} and on the implementation of the Monte Carlo experiment.

\subsection{Estimator specifications}

All estimators target the treatment-specific AH,
\[
\eta_a(\tau) = \frac{F_a(\tau)}{R_a(\tau)}
= \frac{1 - S_a(\tau)}{\int_0^\tau S_a(t)\,dt}, \qquad a \in \{0,1\},
\]
and the log AH ratio $\beta(\tau) = \log\{\eta_1(\tau)/\eta_0(\tau)\}$ at $\tau = 12$ months, where $S_a(t)$ denotes the covariate-adjusted survival curve under treatment~$a$.

\paragraph{\texttt{AH-DML}.}
The estimator \texttt{AH-DML} implements the cross-fitted, doubly robust procedure of Section~\ref{sec:cf-ah} using the \texttt{CFsurvival} package. For each arm $a$, we estimate the marginal survival path $t \mapsto S_a(t)$ via a one-step estimator based on the efficient influence function in Proposition~\ref{prop:eif-mu}, using $K=5$-fold cross-fitting.

Following \citet{westling2024inference}, who highlight the importance of combining parametric, semiparametric, and nonparametric survival learners for nuisance estimation, we use a broad Super Learner library. Parametric accelerated failure time models (exponential, Weibull, log--logistic) supply low-variance baselines. Cox proportional hazards models—with and without treatment--covariate interactions—contribute semiparametric structure. Random survival forests provide a fully nonparametric option with established large-sample properties under right censoring \citep{ishwaran2008rsf, ishwaran2010consistency}. The propensity score is estimated using a flexible Super Learner including generalized linear models and multivariate adaptive regression splines.

Given estimated marginal survival curves $\hat S_a(t)$, we compute
$\hat F_a(\tau) = 1 - \hat S_a(\tau)$,
$\hat R_a(\tau) = \int_0^\tau \hat S_a(t)\,dt$ (via numerical integration),
and $\hat\eta_a(\tau)=\hat F_a(\tau)/\hat R_a(\tau)$, with
$\hat\beta(\tau)=\log\{\hat\eta_1(\tau)/\hat\eta_0(\tau)\}$.
Variance estimates use the empirical variance of the mapped efficient influence function from Theorem~\ref{thm:eif-ah}.

\paragraph{\texttt{G-computation}.}
The \texttt{G-computation} estimator is an outcome-regression plug-in method. A survival Super Learner is used to estimate conditional survival functions $S_a(t \mid W)$; marginal survival is obtained by averaging:
\[
\hat S_a(t)=\mathbb{P}_n\{\hat S_a(t\mid W)\},\qquad a\in\{0,1\}.
\]
These are then mapped to AH and the log AH ratio as above. Inference uses a nonparametric bootstrap with 200 resamples.

\paragraph{\texttt{Cox}.}
Finally, \texttt{Cox} fits a conventional Cox proportional hazards working model,
\[
\lambda(t\mid A,W)=\lambda_0(t)\exp\{\beta_A A+\beta_W^\top W\},
\]
using the \texttt{survival} package. Predicted survival curves are generated under $A=0$ and $A=1$ while holding covariates at their observed values. Marginal survival is obtained by averaging predictions over the empirical distribution of~$W$. AH and the log AH ratio are computed by the same plug-in map. Variance estimation uses a nonparametric bootstrap with $200$ resamples.

\subsection{Simulation workflow and computation}

The simulation grid consists of five sample sizes
($n = 500, 750, 1000, 1250, 1500$), four DGMs (PH-simple, PH-complex, non-PH, and cross-A), and the three estimators described above. For each Monte Carlo replicate $m = 1,\dots,1000$ and each sample size $n$ we:
\begin{enumerate}
    \item generate one dataset under each DGM using Appendix~\ref{app:sim-dgp};
    \item apply all three methods to obtain estimates and confidence intervals for $\eta_0(\tau)$, $\eta_1(\tau)$ and $\beta(\tau)$;
    \item save the results along with metadata (iteration, sample size, DGM, estimator).
\end{enumerate}

\section{Additional simulation results}
\label{app:sim-extra}

\subsection{Non-proportional hazards and crossing-hazards scenarios}
\label{app:sim-crossing}

Figures~\ref{fig:app-nonph-bvm}--\ref{fig:app-cross-coverage} report the full Monte Carlo panels for the non-PH and cross-A DGMs, including arm-specific average hazards and the log AH ratio. Across both DGMs, \texttt{AH-DML} exhibits small bias and near-nominal coverage, while \texttt{Cox} and \texttt{G-computation} can display substantial bias and under-coverage consistent with working-model misspecification.

\begin{figure}[t]
  \centering
  \begin{subfigure}{0.75\textwidth}
    \centering
    \includegraphics[width=\linewidth]{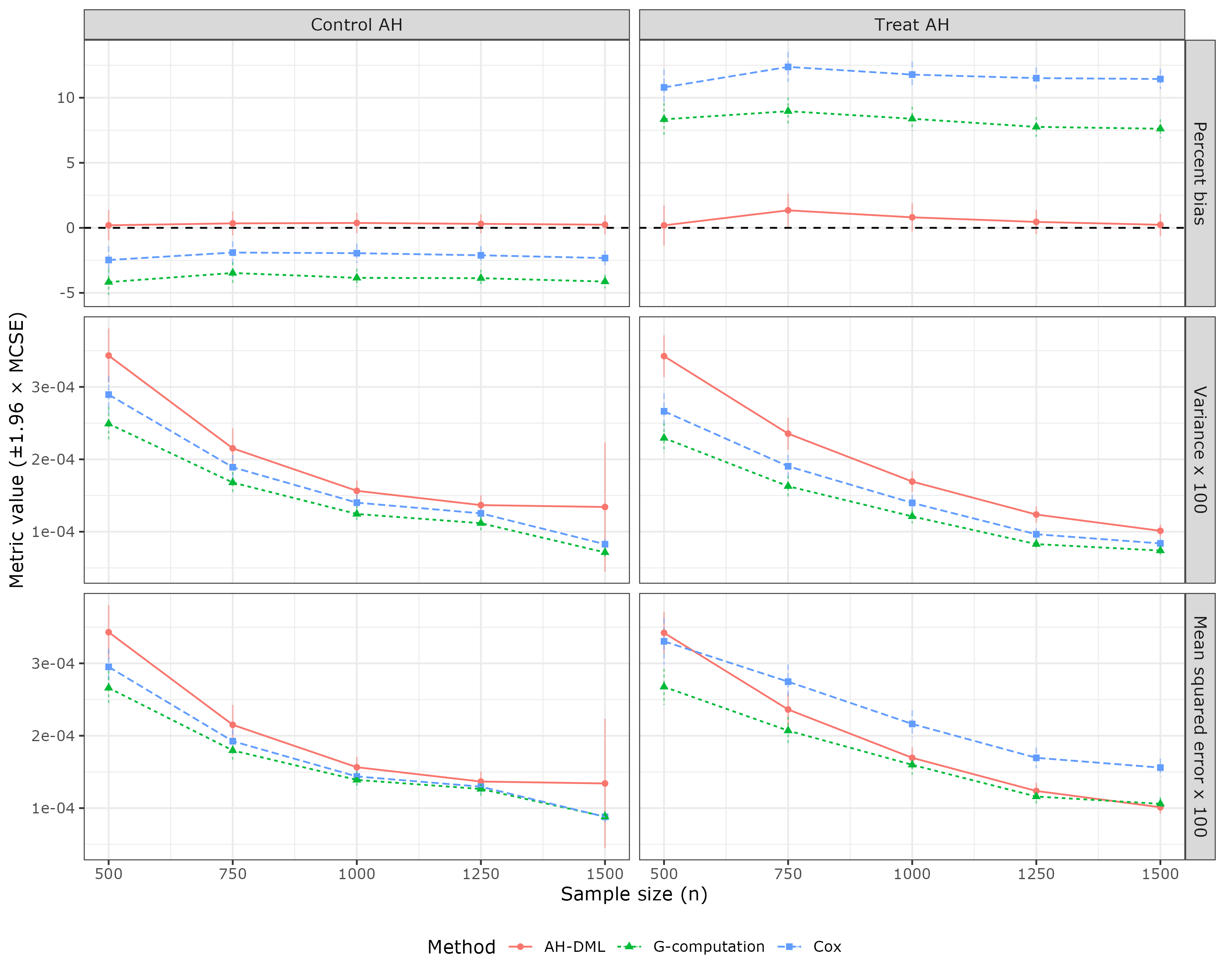}
    \caption{Control and treated AH.}
    \label{fig:app-nonph-ah}
  \end{subfigure}%
  \hfill
  \begin{subfigure}{0.75\textwidth}
    \centering
    \includegraphics[width=\linewidth]{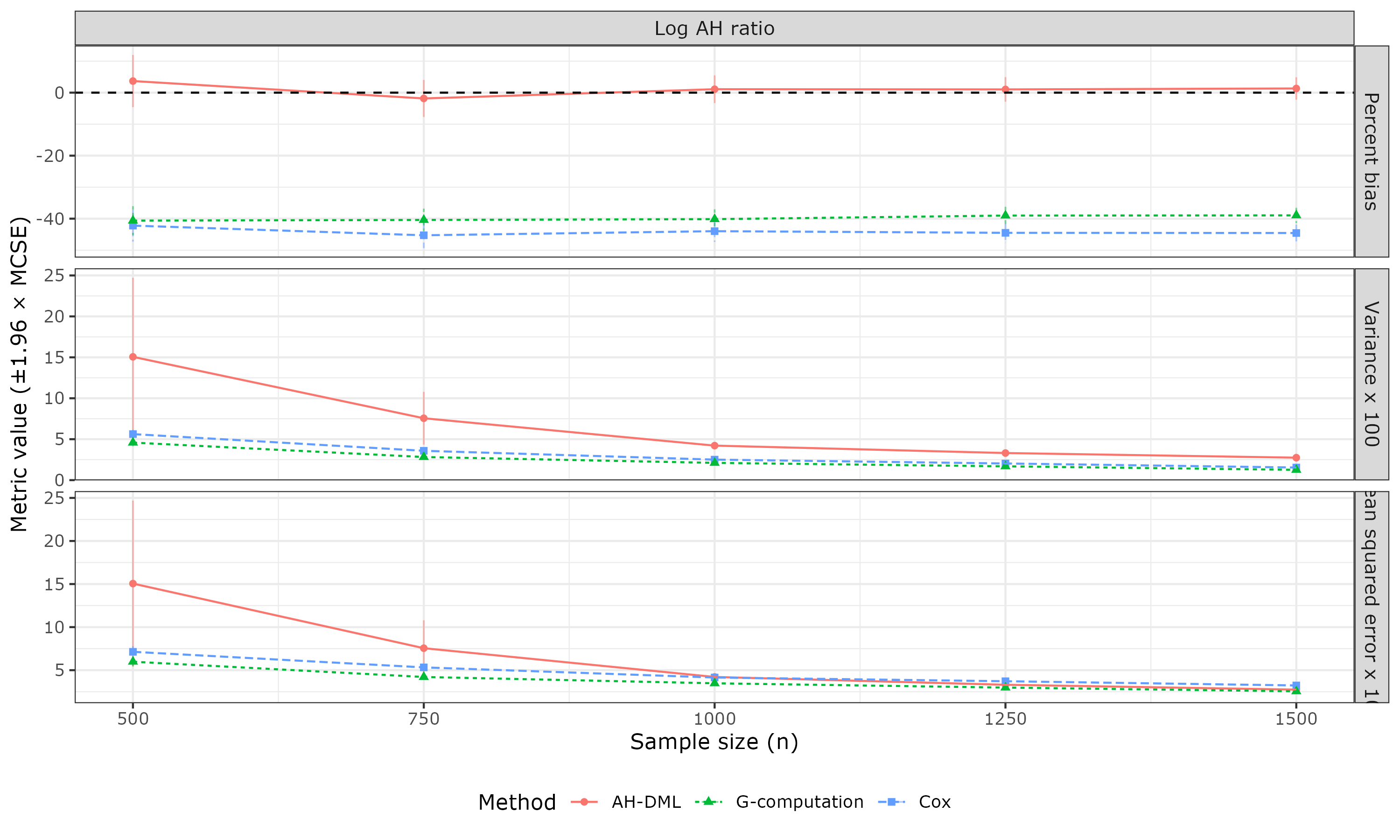}
    \caption{Log AH ratio.}
    \label{fig:app-nonph-ratio}
  \end{subfigure}
  \caption{Non-PH DGM: Monte Carlo percent bias, variance (rescaled by 100) and mean squared error (rescaled by 100).}
  \label{fig:app-nonph-bvm}
\end{figure}

\begin{figure}[t]
  \centering
  \includegraphics[width=0.8\textwidth]{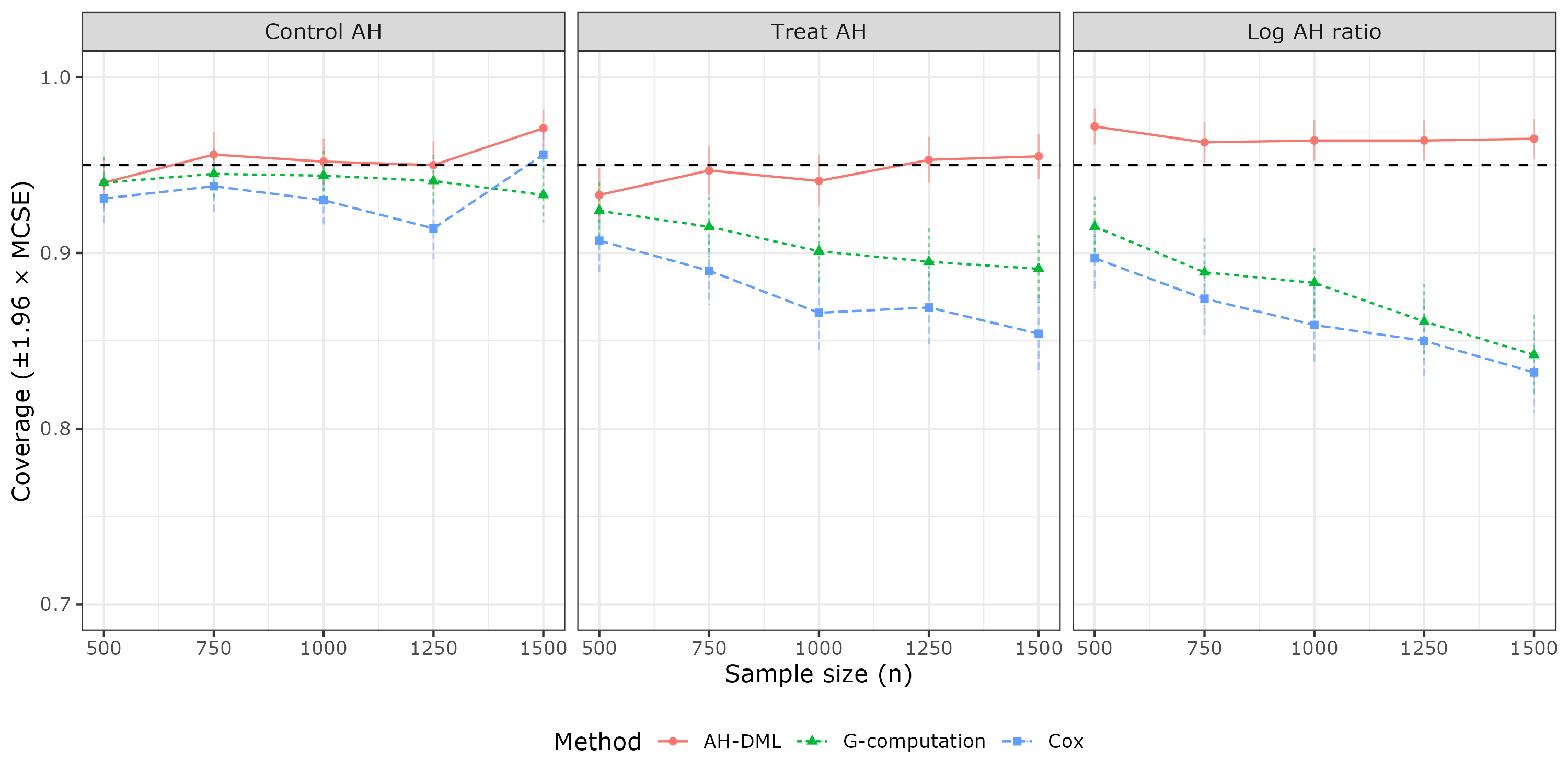}
  \caption{Non-PH DGM: empirical coverage of nominal 95\% confidence intervals. The dashed horizontal line indicates the nominal 0.95 level.}
  \label{fig:app-nonph-coverage}
\end{figure}

\begin{figure}[t]
  \centering
  \begin{subfigure}{0.75\textwidth}
    \centering
    \includegraphics[width=\linewidth]{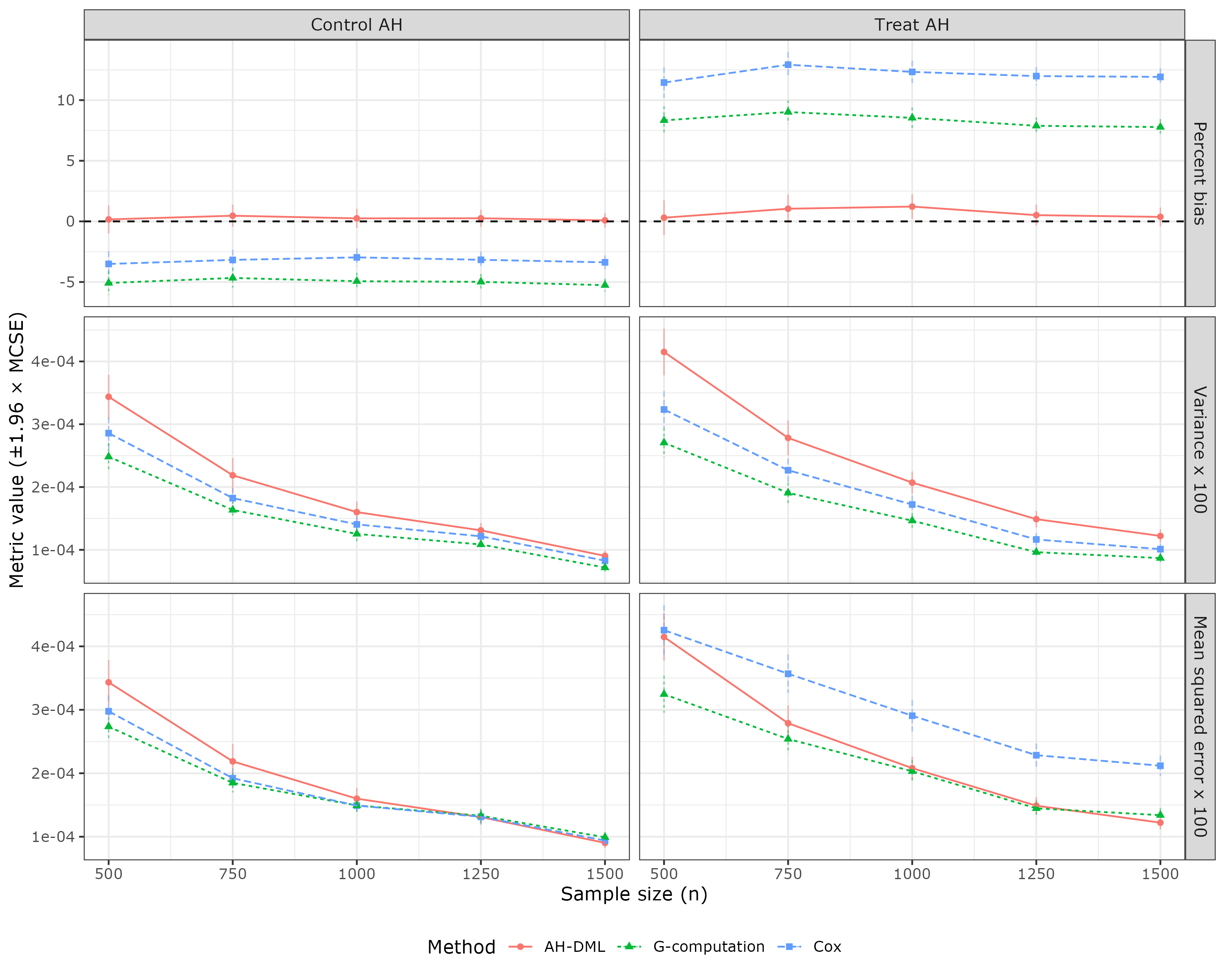}
    \caption{Control and treated AH.}
    \label{fig:app-cross-ah}
  \end{subfigure}%
  \hfill
  \begin{subfigure}{0.75\textwidth}
    \centering
    \includegraphics[width=\linewidth]{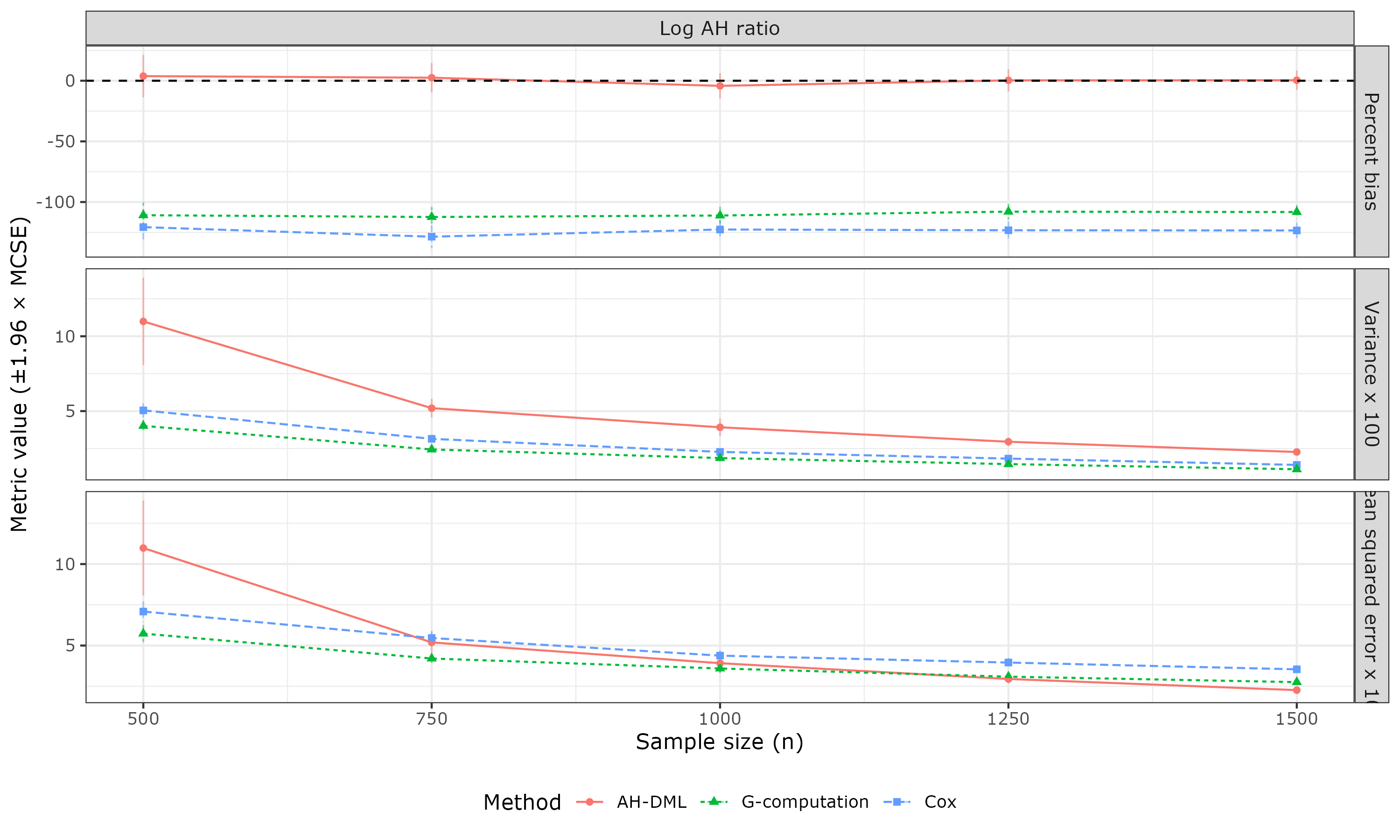}
    \caption{Log AH ratio.}
    \label{fig:app-cross-ratio}
  \end{subfigure}
  \caption{Cross-A DGM: Monte Carlo percent bias, variance (rescaled by 100) and mean squared error (rescaled by 100).}
  \label{fig:app-cross-bvm}
\end{figure}

\begin{figure}[t]
  \centering
  \includegraphics[width=0.8\textwidth]{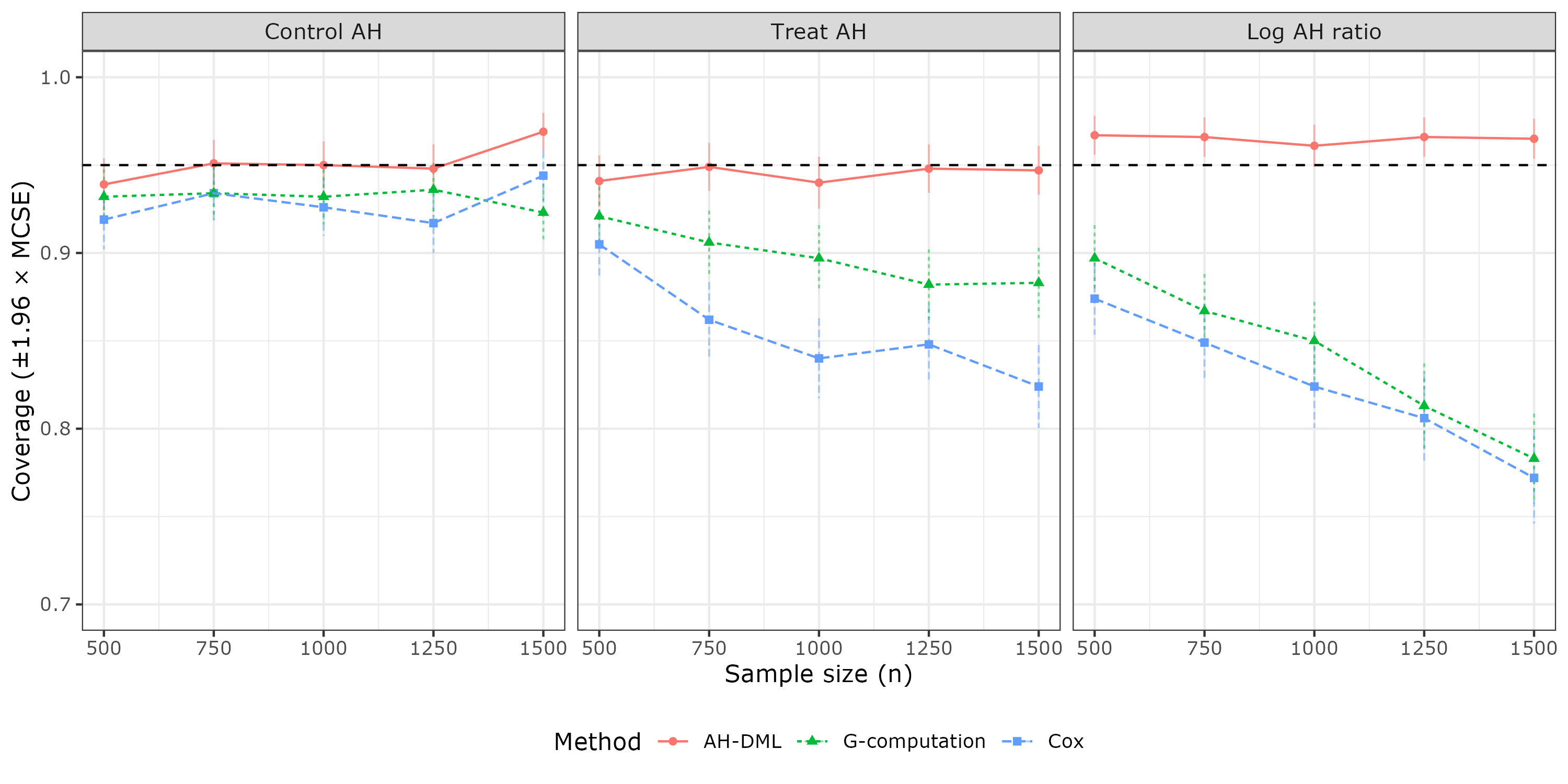}
  \caption{Cross-A DGM: empirical coverage of nominal 95\% confidence intervals. The dashed horizontal line indicates the nominal 0.95 level.}
  \label{fig:app-cross-coverage}
\end{figure}

\subsection{Proportional-hazards scenarios}
\label{app:sim-ph}

We report results for two proportional-hazards settings that share the same Weibull PH event-time mechanism but differ in censoring. The PH setting uses linear exponential censoring (Appendix~\ref{app:sim-ph-linear-dgp}), which is consistently estimable by our censoring library. The PH-complex setting uses a Weibull PH censoring model with rich interaction structure (Appendix~\ref{app:sim-ph-complex-dgp}), which is not consistently estimable by our current library. This comparison illustrates that the practical benefits from double robustness depend on whether at least one nuisance component is well captured by the analyst's learner library.

\begin{figure}[t]
  \centering
  \begin{subfigure}{0.75\textwidth}
    \centering
    \includegraphics[width=\linewidth]{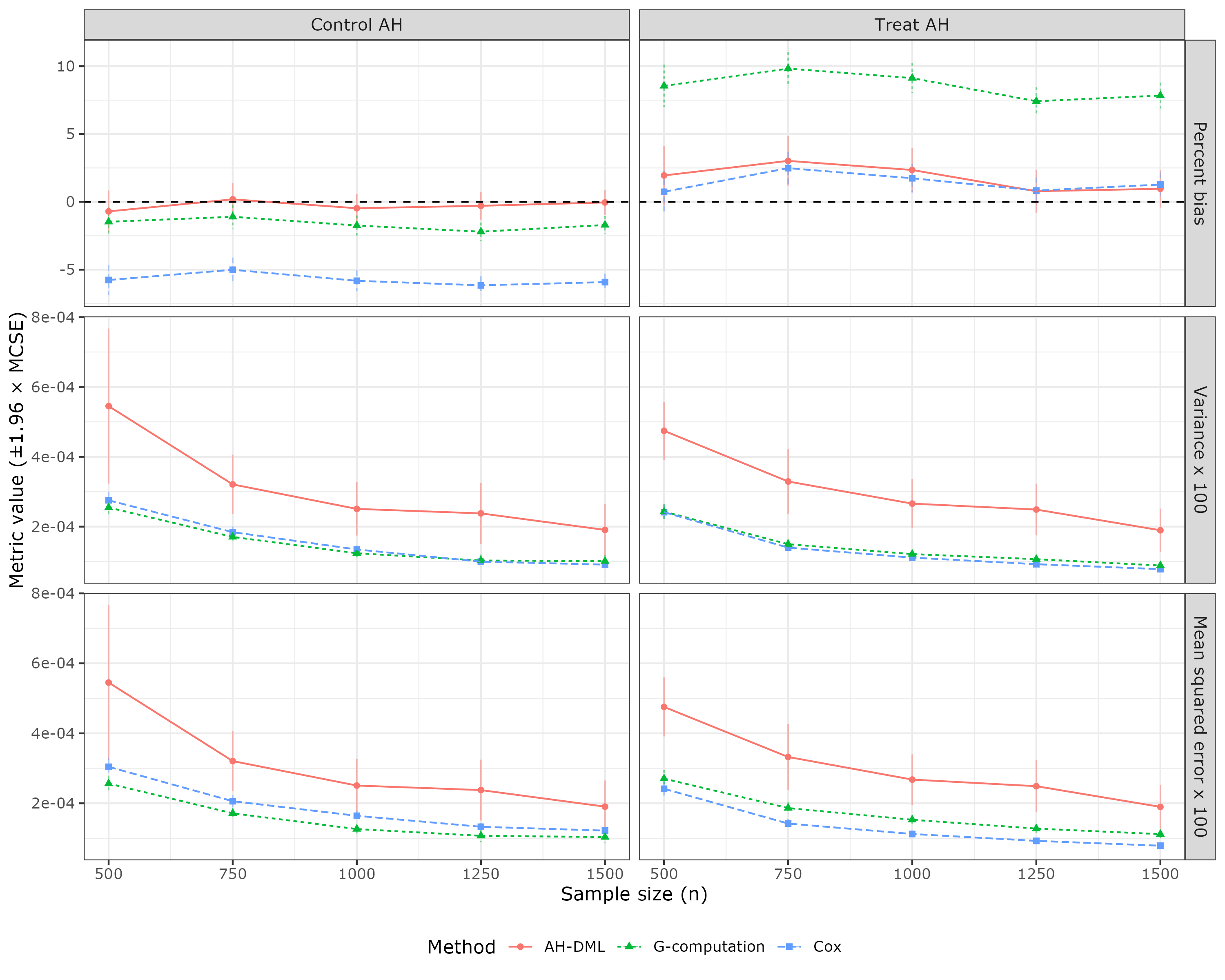}
    \caption{Control and treated AH.}
    \label{fig:app-ph-ah}
  \end{subfigure}%
  \hfill
  \begin{subfigure}{0.75\textwidth}
    \centering
    \includegraphics[width=\linewidth]{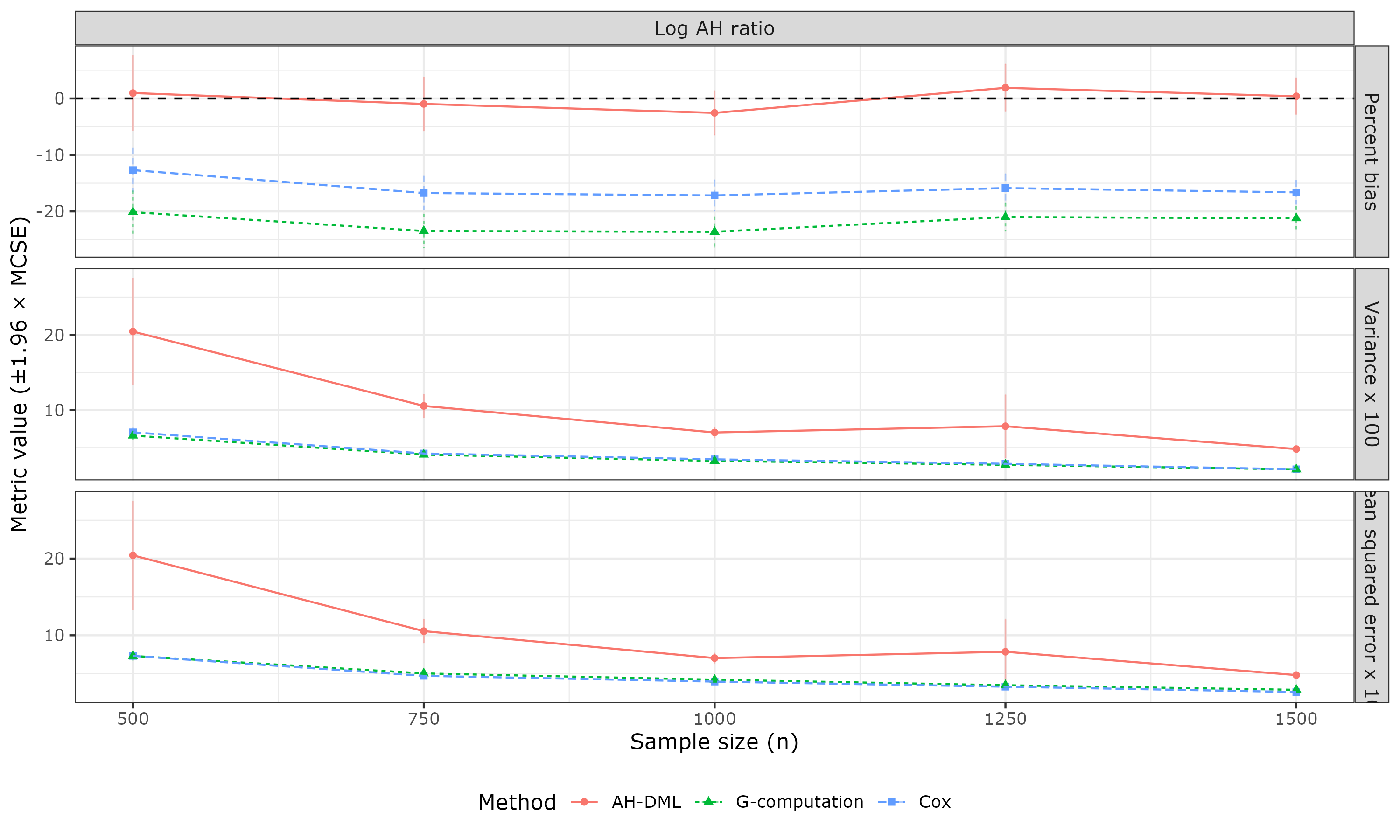}
    \caption{Log AH ratio.}
    \label{fig:app-ph-ratio}
  \end{subfigure}
  \caption{PH-simple DGM (linear censoring): Monte Carlo percent bias, variance (rescaled by 100) and mean squared error (rescaled by 100).}
  \label{fig:app-ph-bvm}
\end{figure}

\begin{figure}[t]
  \centering
  \includegraphics[width=0.8\textwidth]{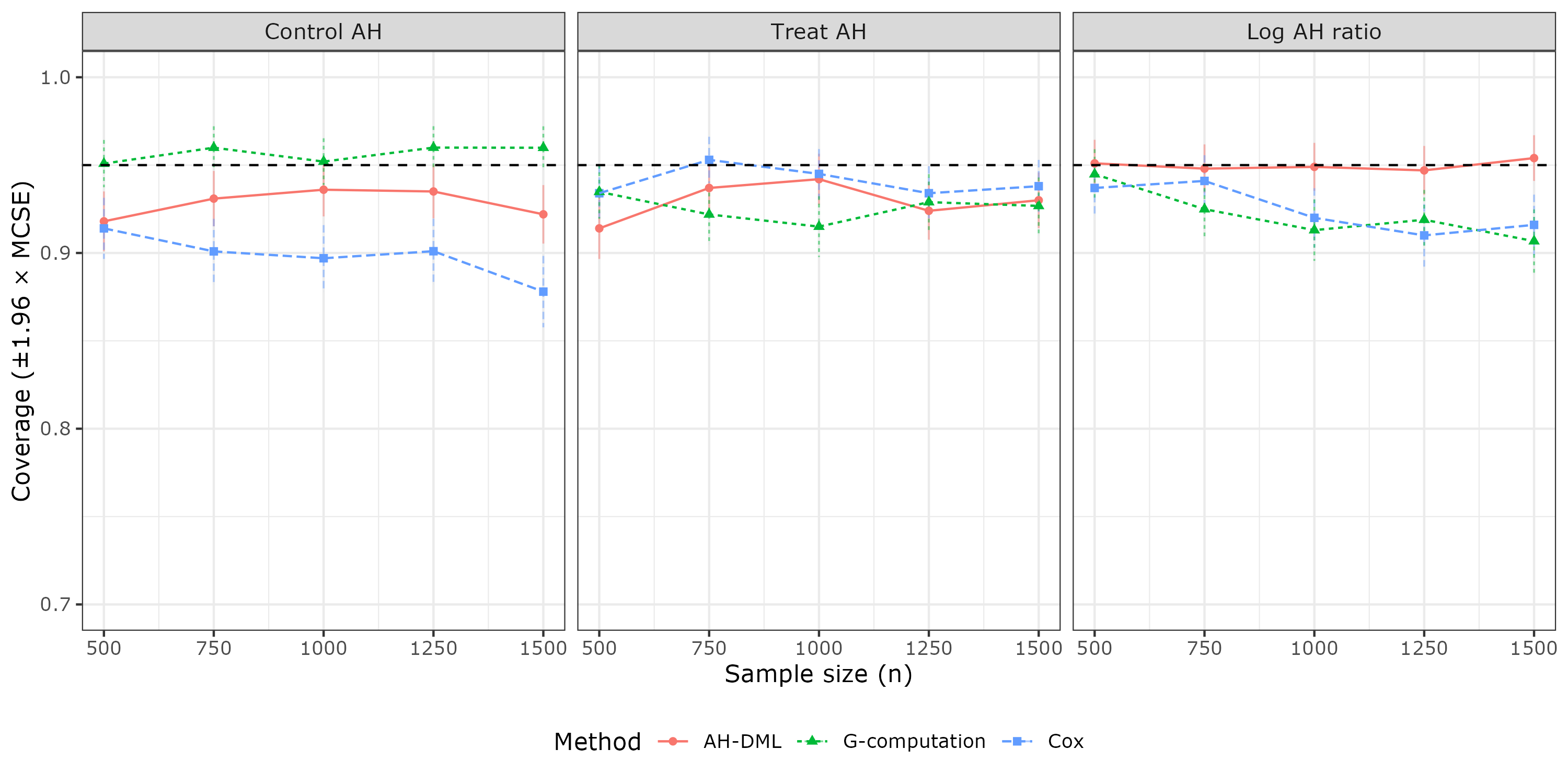}
  \caption{PH-simple DGM: empirical coverage of nominal 95\% confidence intervals. The dashed horizontal line indicates the nominal 0.95 level.}
  \label{fig:app-ph-coverage}
\end{figure}

\begin{figure}[t]
  \centering
  \begin{subfigure}{0.75\textwidth}
    \centering
    \includegraphics[width=\linewidth]{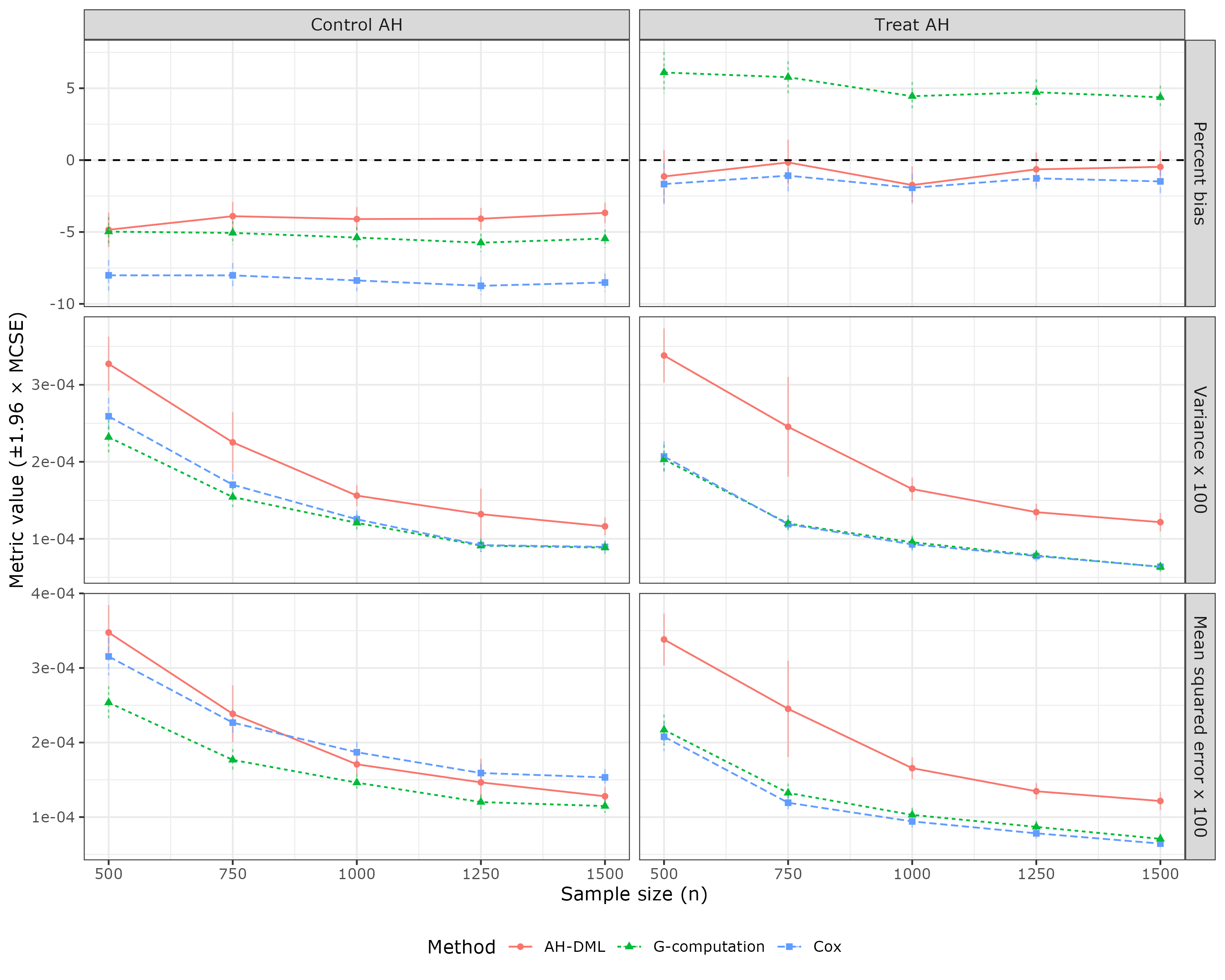}
    \caption{Control and treated AH.}
    \label{fig:app-phc-ah}
  \end{subfigure}%
  \hfill
  \begin{subfigure}{0.75\textwidth}
    \centering
    \includegraphics[width=\linewidth]{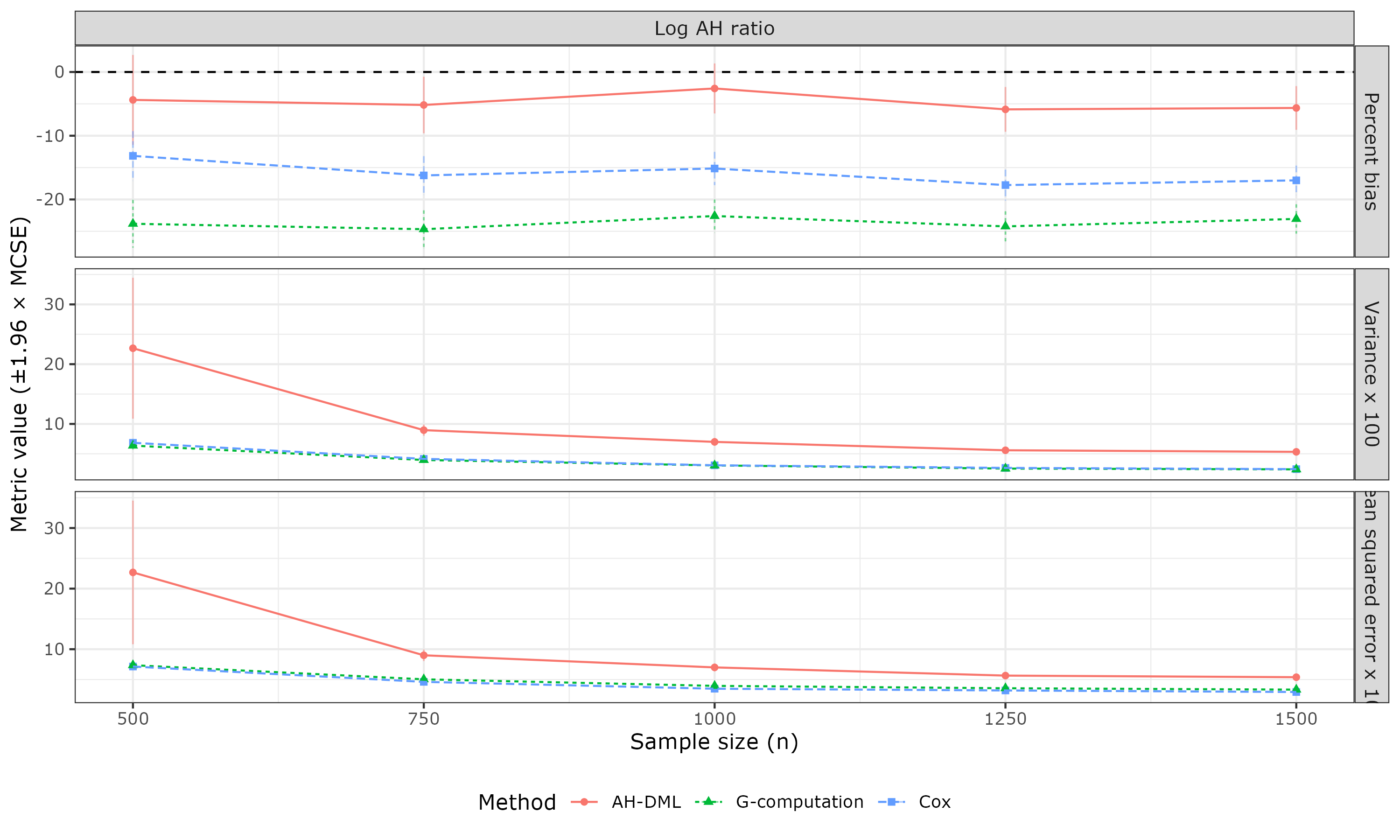}
    \caption{Log AH ratio.}
    \label{fig:app-phc-ratio}
  \end{subfigure}
  \caption{PH-complex DGM: Monte Carlo percent bias, variance (rescaled by 100) and mean squared error (rescaled by 100).}
  \label{fig:app-phc-bvm}
\end{figure}

\begin{figure}[t]
  \centering
  \includegraphics[width=0.8\textwidth]{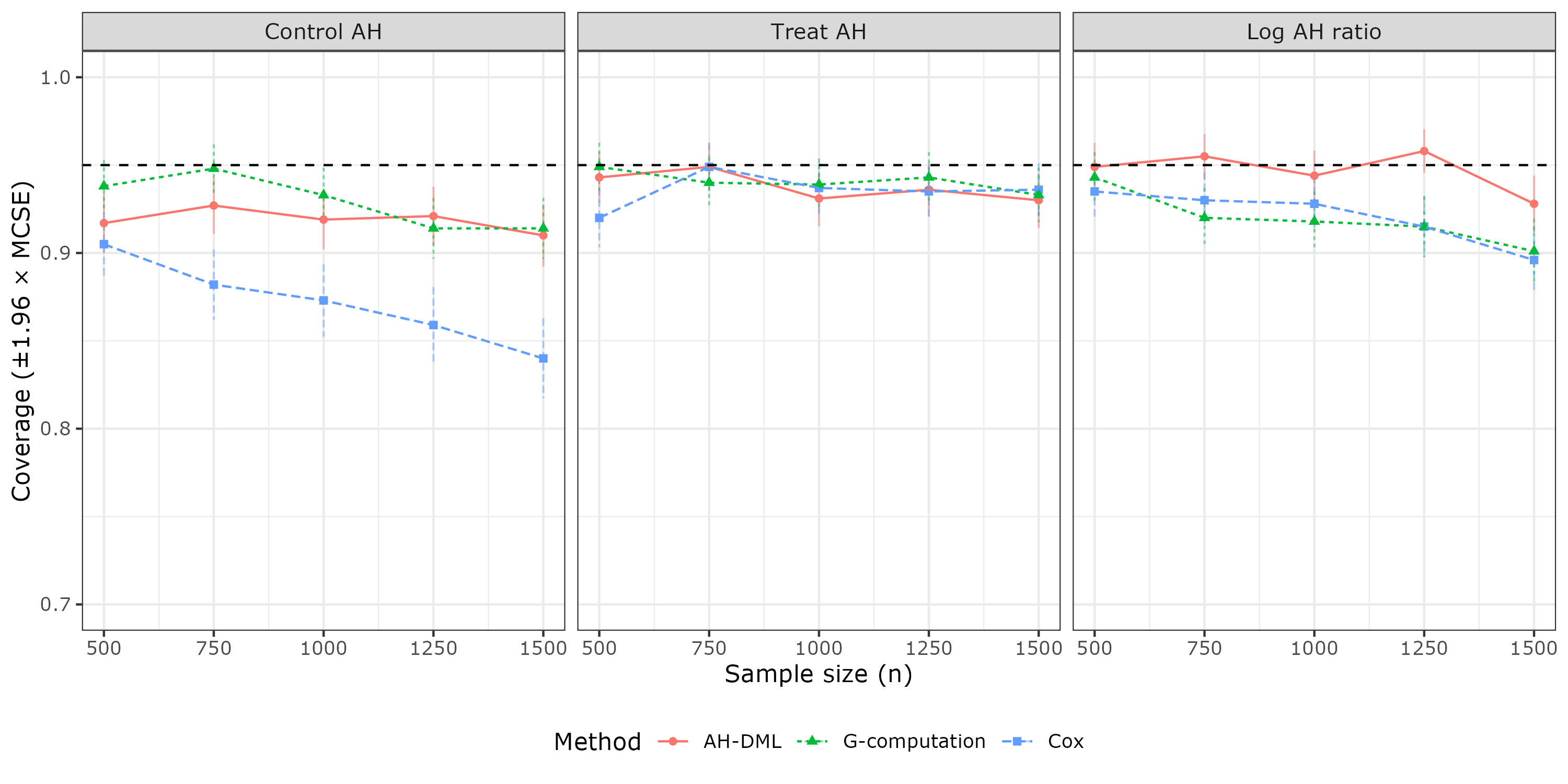}
  \caption{PH-complex DGM: empirical coverage of nominal 95\% confidence intervals. The dashed horizontal line indicates the nominal 0.95 level.}
  \label{fig:app-phc-coverage}
\end{figure}

\end{document}